\documentclass[11pt]{article}

\usepackage[T1]{fontenc}

\usepackage[margin=1in]{geometry}

\usepackage{hyperref}
\hypersetup{
    hidelinks,
    pdfencoding=auto,
}

\RequirePackage{latexsym}
\RequirePackage{amsmath}
\RequirePackage{amssymb}
\RequirePackage{bm}
\RequirePackage[mathscr]{euscript}
\usepackage{cleveref}
\usepackage{amsthm}
\usepackage[normalem]{ulem}
\usepackage{cancel}

\usepackage{color}

\newcommand{\R}{\mathbb{R}} 

\newcommand{\kb}{\mathbf{k}}

\allowdisplaybreaks

\newcommand{\EE}{\mathbb{E}} 
\newcommand{\PP}{\mathbb{P}} 

\newcommand{\floor}[1]{\left\lfloor #1 \right\rfloor}

\newcommand{\norm}[1]{\|#1\|}

\newcommand{\one}{\mathbf{1}}  

\newtheorem{theorem}{Theorem}
\newtheorem{theorem*}{Theorem}
\newtheorem{definition}{Definition}
\newtheorem{lemma}{Lemma}

\newtheorem{corollary}{Corollary}

\usepackage{xcolor}

\title{On the Existence of Fair Allocations for Goods and Chores under Dissimilar Preferences}
\author{Egor Gagushin\thanks{Purdue University. \texttt{eggag03@gmail.com}} 
        \and Marios Mertzanidis\thanks{Purdue University. \texttt{mmertzan@purdue.edu}} 
        \and Alexandros Psomas\thanks{Purdue University. \texttt{apsomas@purdue.edu}}}

\date{}

\usepackage{algorithm}
\usepackage{algorithmic}
\usepackage{mathtools}

\newcommand{\items}{\mathcal{M}}
\newcommand{\agents}{\mathcal{N}}
\newcommand{\normv}{\tilde{v}^{(1)}}
\newcommand{\normc}{\tilde{c}^{(1)}}
\newcommand{\normtwov}{\tilde{v}^{(2)}}
\newcommand{\normtwou}{\tilde{u}^{(2)}}

\newcommand{\gap}{\textsc{Gap}}
\newcommand{\cgap}{\textsc{CGap}}
\newcommand{\trading}{\textsc{TP}}
\newcommand{\invtrading}{\textsc{ITP}}
\newcommand{\mech}{\textsc{RN}}
\newcommand{\choremech}{\textsc{LRN}}

\usepackage{microtype}

\begin{document}

\maketitle

\thispagestyle{empty}
\addtocounter{page}{-1}

\begin{abstract}
We study the fundamental problem of fairly allocating a multiset $\items$ of $t$ types of indivisible items among $d$ groups of agents, where all agents within a group have identical additive valuations. Gorantla et al.~\cite{gorantla2023fair} showed that for every such instance, there exists a finite number $\mu$ such that, if each item type appears at least $\mu$ times, an envy-free allocation exists. Their proof is non-constructive and only provides explicit upper bounds on $\mu$ for the cases of two groups ($d=2$) or two item types ($t=2$).

In this work, we resolve one of the main open questions posed by Gorantla et al.~\cite{gorantla2023fair} by deriving explicit upper bounds on $\mu$ that hold for arbitrary numbers of groups and item types. We introduce a significantly simpler, yet powerful technique that not only yields constructive guarantees for indivisible goods but also extends naturally to chores and continuous domains, leading to new results in related fair division settings such as cake cutting.
\end{abstract}

\tableofcontents
\newpage

\section{Introduction}

Fairly allocating resources is a central problem in computer science and economics. What ``fairness'' means, however, admits multiple formal definitions. A classic notion is \emph{proportionality}: an allocation among $n$ agents is proportional if each agent receives at least a $1/n$ fraction of its total value for all items.  Arguably the gold standard is \emph{envy-freeness}, which asks that no agent prefers another agent's bundle to their own.

Envy-free allocations are not guaranteed to exist in the worst case; consider, for example, a single indivisible item and two agents with a positive value for it. Motivated by such inherent limitations, a substantial body of work tries to identify conditions and settings under which fair allocations become possible. This line of work has culminated in numerous relaxations of envy-freeness, including \emph{envy-freeness up to one item} (EF1) and \emph{envy-freeness up to any item} (EFX); the former is known to exist quite generally, while the latter remains elusive.
In stochastic settings envy-free allocations can emerge with high probability; e.g., when agents have additive preferences, drawn i.i.d.\ from a distribution with a bounded probability density function, Manurangsi and Suksompong~\cite{manurangsi2021closing} show that, for the case of goods, an envy-free allocation exists with high probability when the number of items satisfies $m \in \Omega(n \log n / \log \log n)$.

A complementary, structural route asks whether \emph{multiplicity alone} can guarantee fairness. In this direction, Gorantla et al.~\cite{gorantla2023fair} investigate whether envy-free allocations exist when items are available in multiple identical copies.
Formally, they consider the problem of fairly allocating a multiset $\items$ consisting of $t$ types of indivisible items among $d$ groups of agents, where all agents within a group have identical additive valuations.~\cite{gorantla2023fair} show that, for every such instance, there exists a finite number $\mu$ such that, if each item type has at least $\mu$ copies, then an envy-free allocation exists. Their proof relies on elegant but intricate arguments demonstrating that a certain linear program has a sufficiently large feasibility space to guarantee the existence of an integral feasible solution.

However, this approach has some key limitations: (i) it does not yield an efficient algorithm for finding such allocations when they are guaranteed to exist, (ii) it is too specialized to easily yield results for other fair-division settings, and, more importantly, (iii) it does not provide explicit bounds on $\mu$, as the argument is purely existential by nature. Gorantla et al.~\cite{gorantla2023fair} derive upper bounds on $\mu$ only for the special cases where either the number of item types is two ($t = 2$) or the number of agent groups is two ($d = 2$). 

In this paper, we resolve one of the main questions left open by Gorantla et al.~\cite{gorantla2023fair} by providing explicit bounds on $\mu$ for arbitrary instances and developing a technique that addresses the aforementioned limitations.

\subsection{Our Contributions}

We provide the first \emph{explicit} upper bounds on $\mu$, the number of copies required per item type to guarantee the existence of an envy-free allocation. Specifically, we prove that  
$\mu \in O\!\left( \tfrac{t n^3}{g \delta^2} \right)$, 
where $g$ denotes the greatest common divisor of the agent group sizes, and $\delta$ quantifies the degree of dissimilarity between the valuations of different groups. 
For comparison, Gorantla et al.~\cite{gorantla2023fair} established that, when $d = 2$, one has 
$\mu \in O\!\left( \tfrac{n^2 \sqrt{t}}{g \delta} \right)$ 
and $\mu \in \Omega\!\left( \tfrac{n^2}{g \delta} \right)$, 
while for the case of $t = 2$ they obtained 
$\mu \in O\!\left( \tfrac{n^2}{g \delta} \right)$. 
Interestingly, in the special case where each group consists of a single agent, our bound simplifies to 
$\mu \in O\!\left( \tfrac{n^3}{\delta^2} \right)$, 
showing that, unlike in~\cite{gorantla2023fair}, our result is independent of the number of item types $t$.

At the heart of our analysis is a dissimilarity condition that implies the existence of envy-free allocations. Our condition formalizes a natural intuition: \emph{fairness should be possible if agents' preferences are sufficiently dissimilar}. We quantify this dissimilarity between agents’ valuations using well-established measures of distance, including the Euclidean distance and various $f$-divergences such as the Kullback–Leibler (KL) divergence and the $\chi^2$-divergence.

Our technique is simple and naturally extends across models, yielding results in a broad range of fair-division settings.

\paragraph{Chores and envy-freeness.} We extend our analysis to the setting where items represent \emph{chores}: indivisible tasks that impose costs rather than provide value to agents. By following the same framework developed for the case of goods, we derive explicit upper bounds on $\mu$ that guarantee the existence of envy-free allocations in this setting as well.

\paragraph{Cake Cutting.} We consider the classic envy-free cake cutting problem, in the query model of Robertson and Webb. It is known that there exists a protocol for finding an envy-free allocation that uses at most $O(n^{n^{n^{n^{n^{n}}}}})$ queries~\cite{aziz2016discrete} (the best known lower bound is $\Omega(n^2)$~\cite{procaccia2009thou}). If agents' valuation functions can be represented as polynomials of degree at most $d$, Br{\^a}nzei~\cite{branzei2015note} provides a protocol for finding an envy-free allocation using $O(n^2 \,d)$ queries.
As an application of our condition for the existence of envy-free allocations, we prove that if (i) the agents' valuation functions/densities are $k$-Lipschitz continuous for some constant $k$, and (ii) the $\ell_2$ distance between the valuations of any pair of agents is at least a constant, then we can find a strongly envy-free allocation using $O(n^4)$ queries.

\paragraph{Proportionality for Goods.} Using the same technical framework, we also derive conditions under which proportional allocations exist for both goods and chores. For proportionality, the intuitive interpretation of our condition is that \emph{each agent must be sufficiently far from society’s aggregate valuation}. In the case of goods, society’s valuation for an item is defined as the average of the agents’ $\ell_1$-normalized valuations, and an agent’s distance from society is measured via the $\chi^2$-divergence between her $\ell_1$-normalized valuation vector and society’s valuation vector (see~\Cref{theorem: prop condition} for the formal statement). In the case of chores, society’s cost for an item is given by the harmonic mean of the agents’ $\ell_1$-normalized costs. 

Leveraging the variational characterization of the $\chi^2$-divergence, we show that in a stochastic setting, where each value (or cost) is independently drawn from $U[0,1]$, with high probability all agents are sufficiently far from society when $m \in \Omega(n)$. By combining these results, we recover the state-of-the-art existence guarantees (up to constant factors) for proportional allocations in the stochastic setting.

By combining these results, in this setting we recover, up to a constant factor in $m$, the result of Manurangsi and Suksompong~\cite{manurangsi2021closing} that a proportional allocation of goods exists with high probability when $m \geq n$; for chores, we obtain a weaker guarantee than the $m \in \omega(1)$ bound of~\cite{manurangsi2025asymptotic}.

\subsubsection*{Our Techniques}

The proofs of all aforementioned results follow the same high-level structure.  Here, we outline the main arguments for the case of envy-freeness and goods.

First, we identify fractional mechanisms that admit closed-form guarantees on the envy gap: the worst-case margin by which a group values its own bundle over another's. We introduce a new mechanism, the \textsc{Relative Norm} mechanism (and its counterpart, the \textsc{Log Relative Norm} mechanism, for chores), and show that under this mechanism, the envy gap between agents belonging to different groups is lower bounded by a simple, interpretable quantity that depends on a suitable measure of dissimilarity between the groups.

Second, we formulate a linear program that computes a \emph{fractional} allocation maximizing the minimum envy gap. By definition, the envy gap achieved by the optimal fractional solution is at least as large as that guaranteed by our \textsc{Relative Norm} mechanism. We then develop a rounding procedure that transforms this fractional solution into an integral one while ensuring that (i) each agent receives an integral bundle of items, (ii) all agents within the same group receive identical bundles, and (iii) only a small number of fractional items need to be transferred between groups.

A central technical challenge arises from the requirement that all agents in the same group must receive identical bundles \cite{gorantla2023fair}. This constraint complicates the redistribution of fractional items when the total number of copies to be reallocated cannot be expressed as a nonnegative integer linear combination of group sizes. For example, suppose we need to reassign $11$ copies of an item among two groups of sizes $n_1 = 5$ and $n_2 = 7$. In this case, it is impossible to distribute the items while maintaining identical allocations within each group. To overcome this difficulty, we leverage a result of Brauer~\cite{brauer1942Partitions} for the Frobenius coin problem.

Combining the above observations, we show that if the total value of items transferred from one group of agents to another is smaller than the original envy gap, then the resulting integral allocation remains envy-free. The total transferred value can be upper bounded by the product of (i) the total fraction of items exchanged among groups and (ii) the maximum value assigned to any single item. Consequently, we derive a sufficient condition ensuring envy-freeness of the rounded allocation: if the maximum item value is smaller than a threshold that depends on the pairwise distances between agents’ valuation vectors, then the final integral allocation is envy-free.

Finally, we show that as the number of item copies increases, the normalized value of each individual item decreases, while the pairwise distances between agents’ valuation vectors remain essentially unchanged. Combining these ingredients, we obtain an explicit bound on the number of copies required to ensure that our condition on the maximum item value is satisfied, guaranteeing the existence of an envy-free allocation.

\paragraph{Future Directions.} Beyond the immediate consequences of our work, we believe that our framework provides a promising new lens for understanding the tractability of fairness notions for indivisible items, such as the maximin share guarantee (MMS) and envy-freeness up to any item (EFX). One possibly useful interpretation of our results is that they identify structural features of instances, namely dissimilarity among the agents, that determine when these problems are easy.\footnote{Every proportional allocation is MMS, and every envy-free allocation is EFX, so our results imply the existence of MMS/EFX as well.}
Using this perspective, our results complement well-known tractable cases for MMS and EFX, such as when agents have identical valuations. 

Motivated by this observation, we also explore the problem of allocating indivisible goods to agents that are similar, but not identical. We measure similarity using an $f$-divergence between normalized valuations. As a proof of concept, in~\Cref{thm: tefx} we prove that when the TV distance between the $\ell_1$-normalized valuations of any pair of agents is upper bounded, then a relaxation of EFX called transfer-EFX (tEFX) can be guaranteed for any number of agents.

\subsection{Related Work}

The literature on fair allocations of indivisible goods is vast. Beyond the work of~\cite{gorantla2023fair} that we already discussed, here we briefly mention some other works that are closely related to our paper. We refer the reader to Amanatidis et al.~\cite{amanatidis2023fair} for a recent survey on the allocation of indivisible goods (that also briefly covers some of the results about chores). See Guo et al.~\cite{guo2023survey} for a survey about chores. We also refer the reader to Barbanel~\cite{barbanel2005geometry} for measure-theoretic conditions under which strongly proportional and strongly envy-free allocations exist for goods and chores in a general (cake-cutting) setting.

\paragraph{Relaxations of Proportionality and Envy-Freeness.}
Faced with the non-existence of proportional and envy-free allocations for indivisible items, a natural approach is to consider relaxed fairness notions.
A well-studied relaxation of proportionality is maximin share (MMS) fairness~\cite{budish2011combinatorial}. An allocation is MMS if each agent (weakly) prefers her allocation to the allocation she would have received if she were to split the items into $n$ bundles, and then receives the worst one. Unfortunately, MMS allocations might not exist in the worst case, for goods~\cite{kurokawa2018fair} or chores~\cite{aziz2017algorithms}. Two well-studied relaxations of envy-freeness are envy-freeness up to one item (EF1)~\cite{lipton2004approximately} and envy-freeness up to any item (EFX)~\cite{caragiannis2019unreasonable}. EF1 allocations always exist, even for agents with monotone valuation functions~\cite{lipton2004approximately}, while EFX allocations are only known to exist for $n \leq 3$ additive agents~\cite{chaudhury2024efx,akrami2025efx}; the case of $n \geq 4$ agents is a major open problem in fair division.

For general additive valuations, a long series of works in fair division strives to pin down the exact constant for which $\alpha$-MMS allocations exist (for goods, we know that $\alpha \leq \frac{39}{40}$~\cite{feige2021tight} and $\alpha \geq \frac{10}{13}$~\cite{heidari2026improved}; for chores, we know that $\alpha \leq 13/11$~\cite{huang2023reduction} and $\alpha \geq 44/43$~\cite{feige2021tight}), as well as attack the EFX question from different angles, e.g., approximations (for goods, $0.618$-EFX allocations are known to exist~\cite{amanatidis2020multiple}, and this approximation can be improved in certain cases~\cite{amanatidis2024pushing}; for chores, $2$-EFX is known~\cite{garg2026existence}), or partial allocations (for goods, EFX allocations exist, with at most $n-2$ unallocated items~\cite{mahara2024extension}).

Closer to our work, another major line of work tries to attack the MMS and EFX problems by restricting agents' valuations. For the case of goods, it is known that EFX allocations exist for three types of agents~\cite{hv2025efx} (this implies existence for identical agents), graphs~\cite{Christodoulou2023}, lexicographic preferences~\cite{Hosseini2021}, bi-valued and dichotomous instances~\cite{Babaioff2021,amanatidis2021maximum,garg2023computing,halpern2020fair}, and two types of items~\cite{gorantla2023fair}. 
For the case of chores, EFX allocations are known to exist under leveled preferences~\cite{Gafni2023}, two types of chores~\cite{Aziz2023}, and binary valuations~\cite{tao2025}.
Hummel~\cite{Hummel2023} shows that, for any integer $c > 0$, there exists a number $n_c$, such that an MMS allocation exists for any instance with $n \geq n_c$ agents and at most $n+c$ items, where $n_c \leq \lfloor 0.6597^c \cdot c! \rfloor$ for allocations of goods and $n_c \leq \lfloor 0.7838^c \cdot c! \rfloor$ for allocations of chores. Barman et al.~\cite{Barman2023} prove that MMS allocations for chores exist for binary supermodular cost functions.

\paragraph{Fair allocations in stochastic settings.}
Dickerson et al.~\cite{dickerson2014computational} initiated the study of fair division under stochastic preferences, and proved that maximizing utilitarian welfare produces an envy-free allocation with high probability when the number of goods $m \in \Omega(n \log n)$ and item values are drawn i.i.d. from a fixed distribution with bounded PDF. Interestingly,~\cite{dickerson2014computational} implicitly use the existence of a strongly envy-free fractional allocation (the expected allocation of an agent) to get their result; in an online setting, Benad{\`e} et al.~\cite{benade2024fair} also use a similar approach.

Manurangsi and Suksompong~\cite{manurangsi2021closing} improve the result of~\cite{dickerson2014computational} and show that an envy-free allocation exists with high probability when $m \in \Omega(n \log n/ \log\log n)$. The same work shows that a proportional allocation of goods exists with high probability for $m \geq n$.
Manurangsi and Suksompong~\cite{manurangsi2025asymptotic} study chore allocation in the i.i.d. setting, and prove that, for PDF bounded distributions, envy-free allocations exist with high probability if $m \geq 2n$, and proportional allocations exist with high probability if $m \in \omega(1)$.
Bai and G{\"o}lz~\cite{baienvy} study goods and independent but non-identically distributed additive agents.
Benad{\`e} et al.~\cite{benade2024existence} study envy-free allocations of goods for agents with non-additive stochastic valuations. Beyond envy-freeness and proportionality in stochastic settings, Kurokawa et al.~\cite{kurokawa2016can}, Amanatidis et al.~\cite{amanatidis2017approximation}, and Farhadi et al.~\cite{farhadi2019fair} study the existence of MMS allocations (for goods). We note that all aforementioned works treat the distributions as constants (or are for fixed distributions; e.g., the results of Amanatidis et al.~\cite{amanatidis2017approximation} are for the $U[0,1]$ distribution). Halpern et al.~\cite{halpern2025online} show that envy-free allocations with high probability exist even without this assumption, and even in an online setting.

Finally, for the case of goods, Bai et al.~\cite{bai2022fair} study a smoothed model, where worst-case instances are randomly perturbed, and show that envy-free allocations exist with high probability.

 \section{Preliminaries}\label{sec:prelims}

There is a set $\agents$ of $n$ agents. Each agent belongs to one of $d$ possible groups: group $i$ has $n_i$ agents that have identical, additive preferences. Without loss of generality, we assume that $n_1 \le n_2 \le \cdots \le n_d$. We have that $\sum_{i=1}^d n_i = n$.

Our arguments hinge on the ability to express large integers as nonnegative combinations of the group sizes. Towards this, we rely on the following theorem, due to Brauer~\cite{brauer1942Partitions}, and the ensuing corollary.

\begin{theorem}[\cite{brauer1942Partitions}]
    Let $n_1 \le n_2 \le \cdots \le n_d$, be positive integers such that $\gcd(n_1, \dots, n_d) = 1$. Then, any integer $k \ge (n_1-1)(n_d-1)$ can be written as $k = \sum_{i =1}^d x_i \cdot n_i$, where $x_i\ge0$ is an integer for all $i \in [d]$.  
\end{theorem}

\begin{corollary} \label{corollary: frobenius coin problem}
    Let $n_1 \le n_2 \le \cdots \le n_d$, be positive integers such that $\gcd(n_1, \dots, n_d) = g$. Then, any integer $k \ge g \cdot (n_1/g-1) \cdot (n_d/g-1)$, such that $k \equiv 0 \pmod{g}$, can be written as $k = \sum_{i =1}^d x_i \cdot n_i$, where $x_i\ge0$ is an integer for all $i \in [d]$. 
\end{corollary}

For the remainder of this paper, we will use $g$ for the greatest common divisor of $n_1, \dots, n_d$, and $\theta$ for the threshold from Corollary~\ref{corollary: frobenius coin problem}.

\begin{definition}[gcd $g$ and threshold $\theta$]\label{dfn: theta}
Let  $g \;\coloneqq\; \gcd(n_1,\dots,n_d)$, and $
\theta \;\coloneqq\; g\Big(\tfrac{n_1}{g}-1\Big)\Big(\tfrac{n_d}{g}-1\Big).$
\end{definition}

There is also a multiset $\items$ of $m$ items. There are $t$ types of items, with $k_j \geq 1$ items of type $j$. We have that $\sum_{j=1}^t k_j = m$. We denote by $\kb = [k_1, \dots, k_t]^T$ the vector of the number of copies for each item type.  A practical motivation for this problem is the fair allocation of food donations~\cite{aleksandrov2015online,fooddrop,lee2019webuildai}, where there are many types of items, corresponding to different brands, packages, and so on, and for each type, there are many identical copies.

The items can be either \emph{all goods} or \emph{all chores}. For the case of goods, each group of agents $i \in [d]$ has a vector of values $v_{i}(\kb) \in \R^{m}$ where for each $j \in \items$,  $v_{i,j}(\kb)\ge 0$ denotes the value that each agent in group $i$ has for item $j$. The first $k_1$ entries of $v_{i}(\kb)$ are all equal to the value of an agent in group $i$ for the first type of items, the next $k_2$ entries of $v_{i}(\kb)$ are all equal to the value for the second type of items, and so on. The value of an agent in group $i$ for a subset of items $S \subseteq \items$ is $v_i(\kb,S) = \sum_{j \in S} v_{i,j}(\kb)$. For the case where $\kb = \one$ (i.e., all types of items have one copy), we drop the parentheses and let $v_{i}(\one) = v_{i}$. Throughout, every goods-value vector is nonzero; hence, the normalizations below are valid.

For the case of chores, each group of agents $i \in [d]$ has a vector of costs $c_{i}(\kb) \in \R^{m}$ where for each $j \in \items$,  $c_{i,j}(\kb)> 0$ denotes the cost that each agent in group $i$ has for receiving item $j$. The first $k_1$ entries of $c_{i}(\kb)$ are all equal to the cost for the first type of items, the next $k_2$ entries of $c_{i}(\kb)$ are all equal to the cost for the second type of items, and so on. The cost of an agent in group $i$ for a subset of items $S \subseteq \items$ is $c_i(\kb,S) = \sum_{j \in S} c_{i,j}(\kb)$. Similarly we let $c_i(\one) = c_i$.

Throughout the paper, it will be convenient to work with normalized preferences. For ease of notation, we use $\normv_i(\kb) \in [0,1]^m$ to refer to the $\ell_1$-\emph{normalized} valuation of an agent in group $i$, that is, $\normv_{i,j}(\kb) = \frac{v_{i,j}(\kb)}{v_i(\kb, \items)} = \frac{v_{i,j}(\kb)}{\norm{v_i(\kb)}_1}$ for $j \in \items$. We use notation $\normtwov_i(\kb)$ to refer to the $\ell_2$-\emph{normalized} valuation of an agent in group $i$, that is, $\normtwov_{i,j}(\kb) = \frac{v_{i,j}(\kb)}{\sqrt{\sum_{j \in \items} (v_{i,j}(\kb))^2}} = \frac{v_{i,j}(\kb)}{\norm{v_i(\kb)}_2}$ for $j \in \items$.
Similarly, $\normc_i(\kb) \in [0,1]^m$ is defined as $\normc_{i,j}(\kb) = \frac{c_{i,j}(\kb)}{c_i(\kb, \items)} = \frac{c_{i,j}(\kb)}{\norm{c_i(\kb)}_1}$, for all $j \in \items$; we will not use $\ell_2$-normalized costs (and therefore skip the corresponding notation).

We write $v_{min}(\kb)$ (resp. $c_{min}(\kb)$) and $v_{max}(\kb)$ (resp. $c_{max}(\kb)$) for the minimum and maximum value (resp. cost) of an item. That is, $v_{min}(\kb) = \min_{i \in [d], j \in \items} v_{i,j}(\kb)$, $v_{max}(\kb) = \max_{i \in [d], j \in \items} v_{i,j}(\kb)$, $c_{min}(\kb) = \min_{i \in [d], j \in \items} c_{i,j}(\kb)$, and $c_{max}(\kb) = \max_{i \in [d], j \in \items} c_{i,j}(\kb)$. The $\ell_1$-normalized and $\ell_2$-normalized minimum and maximum values/costs, $\normv_{max}(\kb)$, $\normv_{min}(\kb)$, $\normc_{max}(\kb)$ and so on, are defined analogously.

\paragraph{Allocations.}
A fractional allocation $x \in [0,1]^{d \, m}$ defines, for every agent group $i \in [d]$ and item $j \in \items$, the fraction $x_{i,j} \geq 0$ of the item $j$ that each agent in group $i$ receives. An allocation $x$ is feasible if, for all items $j \in \items$, $\sum_{i \in [d]} n_i \cdot x_{i,j} \leq 1$ and $x_{i,j} \ge 0$ for all $i \in [d]$, $j \in \items$. Throughout the paper, unless explicitly stated, all agents in the same group will receive the exact same bundle of items, and therefore (for the purposes of discussing envy-freeness) it suffices to refer to groups instead of individual agents.
We are interested in \emph{integral} allocations, i.e., each item can be allocated to at most one agent. For ease of notation, we use the notation $A$ for integral allocations and $x$ for fractional allocations. An integral allocation $A=(A_1,\ldots,A_d)$ specifies a multiset of item types $A_i$ received by each agent in group $i\in[d]$. We write $A_{i,z}\in\mathbb{Z}_{\geq0}$ for the number of copies of item type $z$ in $A_i$. Throughout the paper, we only consider \emph{complete} allocations, in which every item is allocated; equivalently, $\sum_{i=1}^d n_iA_{i,z}=k_z$ for every $z\in[t]$.

\paragraph{Fairness Notions.}
An (integral or fractional) allocation of goods $x$ is \emph{proportional} if, for all groups $i \in [d]$, $v_i(\kb, x_i) \geq \frac{v_i(\kb, \items)}{n}$; if the inequality is strict, then the allocation is called strongly proportional.
An (integral or fractional) allocation of goods $x$ is \emph{envy-free} if, for all pairs of distinct groups $i, i' \in [d]$, $v_i(\kb, x_i) \geq v_i(\kb, x_{i'})$; if the inequality is strict, then the allocation is called strongly envy-free.
Let $\gap_{i,i'}(x) = v_i(x_i) - v_i(x_{i'})$ be the advantage (or ``negative envy'') of an agent in group $i$ for an agent in group $i'$, in an allocation $x$ of goods.
An (integral or fractional) allocation of chores $x$ is \emph{proportional} if, for all groups $i \in [d]$, $c_i(x_i) \leq \frac{c_i(\items)}{n}$; if the inequality is strict, then the allocation is called strongly proportional.
An (integral or fractional) allocation of chores $x$ is \emph{envy-free} if, for all pairs of distinct groups $i, i' \in [d]$, $c_i(x_i) \leq c_i(x_{i'})$; if the inequality is strict, then the allocation is called strongly envy-free.
For chores, we similarly define $\cgap_{i,i'}(x) = c_i(x_{i'}) - c_i(x_{i})$ as the advantage of an agent in group $i$ for an agent in group $i'$ in allocation $x$.

\paragraph{Statistical Distances.}

An $f$-divergence $D_f( P || Q)$ measures the distance between two probability measures $P$ and $Q$. For the most part, in this paper, we will consider $f$-divergences for probability mass functions (i.e., discrete distributions), and specifically, we will be studying the $f$-divergence between $\ell_1$-normalized preference vectors.
Let $P = (p_1, \dots, p_z)$ and $Q = (q_1, \dots, q_z)$ be two probability mass functions ($p_i , q_i \geq 0$ for all $i$, and $\sum_{i=1}^z p_i = \sum_{i=1}^z q_i = 1$).
The $\chi^2$-divergence between $P$ and $Q$ is defined as $D_{\chi ^2}(P || Q) \coloneqq \sum_{i=1}^z \frac{(p_i - q_i)^2}{q_i}$. The KL-divergence between $P$ and $Q$ is defined as $D_{KL}(P || Q) \coloneqq \sum_{i=1}^z p_i \ln \left( \frac{p_i}{q_i} \right)$. Finally, the Total variation (TV) distance between $P$ and $Q$ is defined as $D_{TV}(P || Q) \coloneqq \frac{1}{2} \sum_{i=1}^z |p_i - q_i|$.

 \section{Conditions Under Which Envy-Free Allocations Exist}\label{sec: envy freeness}

In this section, we study the existence of envy-free allocations. Missing proofs are deferred to~\Cref{app: missing from main EF section}.
 We start by defining a mechanism for constructing a good \emph{fractional} allocation: the \textsc{Relative Norm} mechanism, given valuations $v_1, \dots, v_d$, outputs an allocation $x^{\mech}$ such that
 \[
 x^{\mech}_{i,j} = \frac{\normtwov_{i,j}(\kb)}{n \, \normtwov_{max}(\kb)} + \frac{1}{n} - \frac{\sum_{i' \in [d]} n_{i'} \cdot  \normtwov_{i',j}(\kb)}{n^2 \cdot \normtwov_{max}(\kb)} 
 \]
 where $x^{\mech}_{i,j}$ is the fraction of item $j \in \items$ that each agent of group $i$ receives, $i \in [d]$. Thus, the mechanism is fractional: $x^{\mech}_{i,j}$ is a share of item $j$. Recall that $\normtwov_{i,j}(\kb)$ is the $\ell_2$-normalized value of an agent in group $i$ for item $j$, and $\normtwov_{max}(\kb)$ is the $\ell_2$-normalized maximum value.

Although the above mechanism may initially appear complicated or arbitrary, it is in fact derived from simple and natural principles. Our goal is to design a mechanism in which the fraction of an item allocated to an agent increases with how much the agent values that item. Since our target property (i.e. envy-freeness) is scale-free, we must employ a normalized notion of value so that agents with uniformly higher valuations across all items do not gain an unfair advantage. So, the most natural idea is to let the allocation scale linearly with the agent’s normalized valuation, i.e., $x_{i,j} = a \cdot \normtwov_{i,j}(\kb) + b$ for some choice of $a, b$. The parameters $a$ and $b$ are chosen to ensure that all fractional allocations are nonnegative and that every item is fully allocated. Following this reasoning, we obtain our \textsc{Relative Norm} mechanism.

First, we argue that this mechanism is feasible.

\begin{lemma}\label{lemma: relative norm feasibility}
The \textsc{Relative Norm} mechanism is feasible. That is, (i) $x^{\mech}_{i,j} \ge 0$ for all $i \in [d]$ and $j \in \items$, and (ii) $\sum_{i \in [d]} n_i \cdot x^{\mech}_{i,j} = 1$  for all $j \in \items$.
\end{lemma}

\begin{proof}
    Towards proving that $x^{\mech}_{i,j} \ge 0$ we have:
    \begin{align*}
        x^{\mech}_{i,j} &=\frac{\normtwov_{i,j}(\kb)}{n \, \normtwov_{max}(\kb)} + \frac{1}{n}   - \frac{\sum_{i' \in [d]} n_{i'} \cdot  \normtwov_{i',j}(\kb)}{n^2 \cdot \normtwov_{max}(\kb)} \\
        &\ge \frac{1}{n} - \frac{1}{n^2} \sum_{i' \in [d]} n_{i'} \cdot  \frac{\normtwov_{i',j}(\kb)}{\normtwov_{max}(\kb)}\\
        &\ge \frac{1}{n} - \frac{1}{n^2} \sum_{i' \in [d]} n_{i'} \\
        &= \frac{1}{n} - \frac{1}{n} = 0
    \end{align*}

    Towards proving that $\sum_{i \in [d]} n_i \cdot  x^{\mech}_{i,j} = 1$ we have:
  
        \begin{align*}
        \sum_{i \in [d]} n_i \cdot x^{\mech}_{i,j} &= \sum_{i \in [d]} n_i \cdot
        \left(\frac{\normtwov_{i,j}(\kb)}{n \, \normtwov_{max}(\kb)} + \frac{1}{n} - \frac{\sum_{i' \in [d]} n_{i'} \cdot  \normtwov_{i',j}(\kb)}
        {n^2 \cdot \normtwov_{max}(\kb)}
        \right) \\
        &= \sum_{i \in [d]} n_i \cdot
        \left( \frac{1}{n}+ \frac{\sum_{i' \in [d]} n_{i'}\left(\normtwov_{i,j}(\kb) - \normtwov_{i',j}(\kb)\right)}
        {n^2 \cdot \normtwov_{max}(\kb)}
        \right) \\
	    &=1 + \frac{1}{n^2 \cdot \normtwov_{max}(\kb)} \sum_{i \in [d]} \sum_{i' \in [d]}  n_i \cdot n_{i'} \cdot \Bigl(\normtwov_{i,j}(\kb) -\normtwov_{i',j}(\kb)\Bigr)\\
        &= 1.
	\end{align*}

\end{proof}

Let $\gap_{i,i'}(x) \coloneqq v_i(x_i) - v_i(x_{i'})$ denote the advantage (or, equivalently, negative envy) of an agent of group $i$ for an agent of group $i'$ under allocation $x$. We call this quantity the \emph{envy gap}; it is the negative of the usual envy $v_i(x_{i'})-v_i(x_i)$. The next lemma gives a lower bound on $\gap_{i,i'}(x^{\mech})$.

\begin{lemma}\label{lem: lower bound on gap}
    For all $i, i' \in [d]$, $\gap_{i,i'}(x^{\mech}) = \frac{\norm{v_i(\kb)}_2 }{2n \cdot \normtwov_{\max}(\kb)} \cdot \norm{\normtwov_i(\kb) - \normtwov_{i'}(\kb)}_2^2$.
\end{lemma}

\begin{proof}
    \begin{align*}
        &\gap_{i,i'}(x^{\mech})
        = \sum_{j \in \items}v_{i,j}(\kb)
           (x^{\mech}_{i,j}-x^{\mech}_{i',j}) \notag\\
        &=\norm{v_i(\kb)}_2\sum_{j\in\items}
          \frac{v_{i,j}(\kb)}{\norm{v_i(\kb)}_2} \\
        &\quad\Biggl(
          \Bigl(\frac{\normtwov_{i,j}(\kb)}
          {n\,\normtwov_{max}(\kb)}+\frac{1}{n}-\frac{\sum_{i''\in[d]}n_{i''}\cdot
          \normtwov_{i'',j}(\kb)}
          {n^2\cdot\normtwov_{max}(\kb)}\Bigr) -\Bigl(\frac{\normtwov_{i',j}(\kb)}
          {n\,\normtwov_{max}(\kb)}+\frac{1}{n} -\frac{\sum_{i''\in[d]}n_{i''}\cdot
          \normtwov_{i'',j}(\kb)}
          {n^2\cdot\normtwov_{max}(\kb)}\Bigr)
          \Biggr)\\
        &=  \norm{v_i(\kb)}_2\sum_{j\in\items}\normtwov_{i,j}(\kb) \left(\frac{\normtwov_{i,j}(\kb)-
          \normtwov_{i',j}(\kb)}
          {n\cdot\normtwov_{max}(\kb)}\right) \\
        &= \frac{\norm{v_i(\kb)}_2}
          {n\cdot\normtwov_{\max}(\kb)}
          \sum_{j\in\items} \left((\normtwov_{i,j}(\kb))^2-
          \normtwov_{i,j}(\kb)\cdot\normtwov_{i',j}(\kb)\right)
        \\
        &= \frac{\norm{v_i(\kb)}_2}
          {n\cdot\normtwov_{\max}(\kb)} \sum_{j\in\items} \Bigl(\frac{1}{2}(\normtwov_{i,j}(\kb))^2
          +\frac{1}{2}(\normtwov_{i',j}(\kb))^2 -\normtwov_{i,j}(\kb)\cdot
          \normtwov_{i',j}(\kb)\Bigr)
        \\
        &= \frac{\norm{v_i(\kb)}_2\cdot \norm{\normtwov_i(\kb)-\normtwov_{i'}(\kb)}_2^2}
          {2n\cdot\normtwov_{\max}(\kb)},\notag
    \end{align*}
where in the penultimate equation we used that $\sum_{j \in \items} (\normtwov_{i,j}(\kb))^2 = \sum_{j \in \items} (\normtwov_{i',j}(\kb))^2 = 1$.
\end{proof}

When $d=1$ and $k_z\equiv0\pmod n$ for every $z\in[t]$,
assigning $k_z/n$ copies of each type to every agent
gives an envy-free allocation.
For the pairwise-gap LPs and bounds below, including
the chores extension, assume $d\ge2$.
Consider the following LP, which finds the fractional allocation that maximizes the minimum $\gap_{i,i'}$, for $\ell_2$-normalized valuations. If an allocation $x$ has a nonnegative objective value with respect to this LP, then $x$ is envy-free.

\begin{equation}\label{LP_gap}
\begin{array}{ll}
\text{maximize} & \alpha \\
\text{subject to}
 & \alpha \le \displaystyle\sum_{j \in \items}
   \normtwov_{i,j}(\kb)\cdot(x_{i,j}-x_{i',j}) \\
 & \hfill \forall i,i'\in[d],\ i'\neq i \\
 & \displaystyle\sum_{i=1}^{d}n_i\cdot x_{i,j}=1
   \quad \forall j\in\items \\
 & \alpha,x_{i,j}\geq0
   \quad \forall i\in[d],\ \forall j\in\items
\end{array}
\end{equation}

The next lemma argues about the existence of an integral allocation whose minimum $\gap$ is not too far from the optimal value of LP~\eqref{LP_gap}.

\begin{lemma}\label{lemma: rounding envy goods}
If $k_z \ge \theta$ and $k_z \equiv 0 \pmod{g}$ for all $z \in [t]$ then there exists an integral allocation $A$ such that, for all $i \in [d]$, $\min_{ i' \in [d],\ i'\ne i} \gap_{i,i'}(A) \geq \norm{v_i(\kb)}_2 \cdot \alpha^* - 2(d(d-1) +t(\theta+n + n_d - d -1))\max_{j \in \items} v_{i,j}(\kb)$, where $\alpha^*$ is the optimal value of LP~\eqref{LP_gap}.
\end{lemma}

\begin{proof}

    Choose an optimal basic feasible solution of LP~\eqref{LP_gap}. Since the LP has $dm+1$ variables, this solution has at least $dm+1$ tight constraints. There are $d(d-1) + m +  d \, m + 1$ constraints, out of which $d \, m $ constraints are nonnegativity constraints for $x_{i,j}$. By simple counting, we have that at least $d \, m + 1 - (d(d-1) + m + 1)$ out of these nonnegativity constraints must be tight. Therefore, at most $m + d(d-1)$ of the ``$x_{i,j} \geq 0$'' are \emph{not} tight. Since, without loss of generality, for every $j \in \items$, there exists some $i \in [d]$ such that $x_{i,j}>0$ (otherwise, we could set $x_{i,j} = 1/n$ for all $i \in [d]$ without affecting the objective), at most $d(d-1)$ items are shared between groups of agents.

    Let $x^*$ be an optimal solution of LP~\eqref{LP_gap} where at most $d(d-1)$ items are shared between groups. 
    For each group $i$ and item type $z \in [t]$, let $B_{i,z} = n_i\sum_{j \in \items: \text{$j$ has type $z$}} x^*_{i,j}$ be the total mass of type-$z$ items allocated to group $i$ under $x^*$.
    We will make $B_{i,z}$s integral, without changing them too much. 

    For each $z\in[t]$, initialize $S_z \leftarrow 0$. We will have three phases of
    decreasing $B_{i,z}$s and increasing the corresponding $S_z$. After these phases, we will guarantee that for every $i \in [d]$ and $z \in [t]$, $B_{i,z} \equiv 0\pmod{n_i}$. Therefore, 
    the $B_{i,z}$ items of type $z$ can be distributed integrally and equally among the agents in group $i$, for every $i \in [d]$.
    We will also guarantee that $S_z = \sum_{i \in [d]} y_{i,z} \cdot n_i$ for some set of nonnegative integers $\{y_{i,z}\}_{i \in [d]}$; this implies that
    the $S_z$ items of type $z$ can be integrally distributed among the groups so that all agents in the same group receive the same number of items.

    \emph{Phase 1.} 
    Consider every item $\ell \in \items$ such that there exist $i, i' \in [d]$ (with $i \neq i'$) with $x^*_{i,\ell}>0$ and $x^*_{i',\ell}>0$; that is, $\ell$ is an item that is fractionally shared among at least two groups of agents.
    For every such item $\ell$, set $B_{i,z} \leftarrow B_{i,z} - n_i x^*_{i,\ell}$ for all $i \in [d]$, and $S_z \leftarrow S_z + \sum_{i \in [d]}n_i x^*_{i,\ell}=S_z+1$, where $z$ is the type of item $\ell$.
    From our discussion above, we have that there are at most $d(d-1)$ such items $\ell$. Therefore, at the end of Phase 1, we have that $\sum_{z \in [t]} S_z \le d(d-1)$.

    \emph{Phase 2.}
    At the start of Phase 2, $B_{i,z}$ is a nonnegative integer for every $i\in[d]$ and $z\in[t]$: after the shared items have been removed, each remaining item is assigned only to one group $i$ and hence contributes $n_i x^*_{i,j}=1$ to $B_{i,z}$.
    Let $B_{i,z} = n_i \cdot q_{i,z} +r_{i,z}$, where $r_{i,z} \in [0, n_i-1]$. 
    For every $z \in [t]$, set $B_{i,z} \leftarrow B_{i,z} - r_{i,z}$ for every $i \in [d]$, and set $S_z \leftarrow S_z + \sum_{i \in [d]} r_{i,z}$. Therefore, at the end of Phase 2, we have that $B_{i,z} = n_i \cdot q_{i,z}$ for some integer $q_{i,z}$, and that $\sum_{z \in [t]} S_z \le d(d-1) +  t\sum_{i \in [d]}(n_i-1) = d(d-1) + t(n-d)$.

    \emph{Phase 3.}  For every $z \in [t]$ while $S_z < \theta$, we pick an arbitrary group of agents $i \in [d]$, such that $B_{i,z} > 0$, and 
    set $B_{i,z} \leftarrow B_{i,z} - n_i$ and $S_z \leftarrow S_z + n_i$. First, note that it cannot be that $B_{i,z} = 0$ for all $i \in [d]$ and $S_z < \theta$, since $k_z = \sum_{i\in [d]} B_{i,z} + S_z \geq \theta$; therefore, if $S_z < \theta$ then $B_{i,z} > 0$. Moreover, 
    at the end of Phase 2, we are guaranteed that $B_{i,z} = n_i \cdot q_{i,z}$, for some nonnegative integer $q_{i,z}$; therefore $B_{i,z} \leftarrow B_{i,z} - n_i$ is a valid operation. Repeating until $S_z \geq \theta$ for all $z \in [t]$, we have that
    $\sum_{z \in [t]} S_z \le d(d-1) + t(n-d)+t(\theta+n_d-1)$.
    We also have that, for every $i\in[d]$ and $z\in[t]$, $B_{i,z}=n_iq'_{i,z}$ for some nonnegative integer $q'_{i,z}$. Hence $\sum_{i\in[d]}B_{i,z}=\sum_{i\in[d]}n_iq'_{i,z}\equiv0\pmod g$.
    But since $k_z \equiv 0 \pmod{g}$ and $S_z + \sum_{i \in [d]} B_{i,z} = k_z$, we also have that $S_z \equiv 0 \pmod{g}$ for all $z \in [t]$. 
    
    Since $S_z \equiv 0 \pmod{g}$ and $S_z \geq \theta$, applying $\Cref{corollary: frobenius coin problem}$ we have that there exists a set of nonnegative integers $\{y_{i,z}\}_{i \in [d]}$ such that $S_z = \sum_{i \in [d]} y_{i,z} \cdot n_i$.
    We can therefore distribute these $S_z$ items by giving each agent in group $i$ exactly $y_{i,z}$ items of type $z$, for every $i\in[d]$ and $z\in[t]$.
    Likewise, since $B_{i,z}=n_iq'_{i,z}$, we can distribute these $B_{i,z}$ items by giving each agent in group $i$ exactly $q'_{i,z}$ items of type $z$.
    Let $A$ be this allocation, and notice that $A$ allocates all items.
    We have that $\sum_{z \in [t]} y_{i,z} \le \sum_{z \in [t]} S_z \le d(d-1) + t(n-d)+t(\theta+n_d-1)$. Thus, after Phase 3, every agent has received at most $d(d-1) + t(n-d)+t(\theta+n_d-1)$ extra items (compared to the allocation $x^*$).
    All agents of the same group have received exactly the same integral bundle. Also, compared to the allocation $x^*$, each agent in group $i$ has lost at most $(d(d-1) +t(\theta+n + n_d - d -1))\max_{j \in \items} v_{i,j}(\kb)$ value and each agent in any other group has gained at most that amount in value. Thus $\min_{ i'\ne i} \gap_{i,i'}(A) \ge \min_{ i'\ne i} \gap_{i,i'}(x^*) -2(d(d-1)+t(\theta+n+n_d-d-1))\max_{j\in\items}v_{i,j}(\kb) \ge\norm{v_i(\kb)}_2\alpha^* -2(d(d-1)+t(\theta+n+n_d-d-1))\max_{j\in\items}v_{i,j}(\kb)$.
\end{proof}

Given these ingredients, we can prove that an envy-free allocation exists if we have enough items of each type, and the maximum $\ell_2$-normalized values are small.

\begin{theorem} \label{theorem: Conditions for envy freeness}
    If $k_z \ge \theta$ and $k_z \equiv 0 \pmod{g}$, for all $z \in [t]$, and $$\normtwov_{max}(\kb) \le \sqrt{\frac{\min_{i \in [d],i' \in [d],\ i\ne i'} \norm{\normtwov_i(\kb) - \normtwov_{i'}(\kb)}^2_2}{4n(d(d-1) +t(\theta+n + n_d - d -1))} },$$ then an envy-free allocation exists.
\end{theorem}

\begin{proof}
Let $\alpha^*$ be the optimal value of LP~\eqref{LP_gap}; \Cref{lemma: rounding envy goods} implies that there exists an integral solution $A$ such that, for all $i \in [d]$ , $\min_{ i' \in [d],\ i'\ne i} \gap_{i,i'}(A) \geq \norm{v_i(\kb)}_2 \cdot \alpha^* - 2(d(d-1) +t(\theta+n + n_d - d -1))\max_{j \in \items} v_{i,j}(\kb)$.
$x^{\mech}$ is a feasible solution to LP~\eqref{LP_gap}, therefore, $\alpha^* \geq \min_{\hat{i}, \tilde{i} \in [d],\ \hat{i}\ne\tilde{i}} \frac{\gap_{\hat{i}, \tilde{i}}(x^{\mech})}{\norm{v_{\hat{i}}(\kb)}_2}$, which, by~\Cref{lem: lower bound on gap}, is at least $\min_{\hat{i}, \tilde{i} \in [d],\ \hat{i}\ne\tilde{i}} \frac{\norm{\normtwov_{\hat{i}}(\kb) - \normtwov_{\tilde{i}}(\kb)}_2^2}{2n \cdot \normtwov_{\max}(\kb)}$. Putting everything together, we have that for all $i \in [d]$, $\min_{i' \in [d],\ i'\ne i} \gap_{i,i'}(A) \geq \norm{v_i(\kb)}_2\min_{\hat{i}, \tilde{i} \in [d],\ \hat{i}\ne\tilde{i}} \frac{\norm{\normtwov_{\hat{i}}(\kb) - \normtwov_{\tilde{i}}(\kb)}_2^2}{2n \cdot \normtwov_{\max}(\kb)} -  2(d(d-1) +t(\theta+n + n_d - d -1)) \max_{j \in \items} v_{i,j}(\kb)$. Therefore, if, for all $i \in [d]$, $\max_{j \in \items} \normtwov_{i,j}(\kb) \leq \min_{\hat{i}, \tilde{i} \in [d],\ \hat{i}\ne\tilde{i}} \frac{\norm{\normtwov_{\hat{i}}(\kb) - \normtwov_{\tilde{i}}(\kb)}_2^2}{4n(d(d-1) +t(\theta+n + n_d - d -1)) \cdot \normtwov_{\max}(\kb)}$, then $A$ is envy-free (since all $\gap_{i,i'}(A)$ are nonnegative).
Re-arranging and simplifying, we get the theorem.
\end{proof}

As a corollary to~\Cref{theorem: Conditions for envy freeness}, we can get the following simpler expression for the case where every group has only one agent (i.e., $d=n$, or $n_i = 1$ for all $i$), by observing that $\theta = 0$ in this case.

\begin{corollary} \label{corollary: Conditions for envy freeness with single agents groups}
    If $d=n$ and $\displaystyle \normtwov_{max}(\kb) \le \sqrt{\frac{\min_{i,i' \in [d],\ i\ne i'} \norm{\normtwov_i(\kb) - \normtwov_{i'}(\kb)}^2_2}{4n^3} }$, then an envy-free allocation exists.
\end{corollary}

In the remainder of this section, we aim to prove a version of~\Cref{theorem: Conditions for envy freeness} that removes the second condition on the maximum $\ell_2$-normalized value, and only lower bounds the number of items of each type.

We start with a couple of technical lemmas.

\begin{lemma}\label{lemma: copies lemma goods}
    For any pair $i, i' \in [d]$, $\norm{\normtwov_i(\kb) - \normtwov_{i'}(\kb)}_2 \ge \sqrt{\frac{\underline{k}}{\overline{k}}} \cdot \norm{\normtwov_i - \normtwov_{i'}}_2 - \left(1-\sqrt{\underline{k}/\overline{k}}\right)$ where $0<\underline{k}\le\overline{k}$ and $k_z\in[\underline{k},\overline{k}]$ for each $z\in[t]$.
\end{lemma}

\begin{lemma}\label{lemma: function lower bound}
    $x\left( \sqrt{\frac{x}{x+a}} \cdot b - \left(1-\sqrt{\frac{x}{x+a}}\right)\right)^2 \ge b^2x - ab(b+1)$, where $a\ge0$, $x,b>0$, and $x\ge \frac{a(b+1)}{2b}$.
\end{lemma}

We are now ready to prove our main result.

\begin{theorem}\label{application: goods copies}
    Assume $d\ge2$ and $\min_{i,i'\in[d],\,i\ne i'} \norm{\normtwov_i-\normtwov_{i'}}_2^2>0$. If $k_z \ge \frac{4n(d^2 +t(\theta+n + n_d - d -1))}{ \min_{i,i' \in [d],\ i\ne i'}\norm{\normtwov_i - \normtwov_{i'}}^2_2}$ and $k_z \equiv 0 \pmod{g}$ for all $z \in [t]$, an envy-free allocation exists.
\end{theorem}

\begin{proof}

    The case $d=1$ is immediate, so assume $d\ge2$ for the remainder of the proof.

    Let $\mu = \left\lceil\frac{4n(d^2 +t(\theta+n + n_d - d -1))}{ \min_{i \in [d], i' \in [d],\ i\ne i'}\norm{\normtwov_i - \normtwov_{i'}}^2_2}\right\rceil$. First, notice that, without loss of generality, we can assume that $k_z \in [\mu, \mu+n-1]$. If that is not the case, and $k_z > \mu+n-1$, we can simply allocate one copy of item $z$ to each agent and continue this process until the condition is satisfied; since all agents receive the exact same amount of items, envy is not affected by this process. Also, since we allocate copies in multiples of $n$ and $n \equiv 0 \pmod{g}$, then the remaining number of copies is also divisible by $g$ (so, the condition $k_z \equiv 0 \pmod{g}$ also remains true).

    From \Cref{theorem: Conditions for envy freeness} we know that if $\normtwov_{max}(\kb) \le \sqrt{\frac{\min_{i \in [d],i' \in [d],\ i\ne i'} \norm{\normtwov_i(\kb) - \normtwov_{i'}(\kb)}^2_2}{4n(d(d-1) +t(\theta+n + n_d - d -1))} }$, then an envy-free allocation exists. Notice that:
    \begin{align*}
    \normtwov_{max}(\kb)
     &= \max_{i \in [d]}
        \frac{\max_{j \in \items}v_{i,j}(\kb)}{\norm{v_i(\kb)}_2} \\
     &= \max_{i \in [d]}
        \frac{\max_{z \in [t]}v_{i,z}}{\norm{v_i(\kb)}_2} \\
     &\le \max_{i \in [d]}
        \frac{\max_{z \in [t]}v_{i,z}}{\sqrt{\mu}\norm{v_i}_2} \\
     &\le \frac{\normtwov_{max}}{\sqrt{\mu}} \\
     &\le \frac{1}{\sqrt{\mu}}.
    \end{align*}
    Thus, if $\mu \cdot \min_{i \in [d],i' \in [d],\ i\ne i'} \norm{\normtwov_i(\kb) - \normtwov_{i'}(\kb)}^2_2 \ge  4n(d(d-1) +t(\theta+n + n_d - d -1))$ an envy-free allocation exists. We also have that for every $i, i' \in [d]$ with $i\ne i'$:
    \begin{align*}
        0 &= \frac{4n(d^2+t(\theta+n+n_d-d-1))}
	        {\min_{i\in[d],i'\in[d],\ i\ne i'}
	        \norm{\normtwov_i-\normtwov_{i'}}^2_2} \cdot\min_{i\in[d],i'\in[d],\ i\ne i'}
	        \norm{\normtwov_i-\normtwov_{i'}}^2_2 -4n(d^2+t(\theta+n+n_d-d-1)) \\
        &\le  \mu\cdot\min_{i\in[d],i'\in[d],\ i\ne i'}
        \norm{\normtwov_i-\normtwov_{i'}}^2_2 -4n(d^2+t(\theta+n+n_d-d-1)) \\
        &\le  \mu\cdot\min_{i\in[d],i'\in[d],\ i\ne i'}
        \norm{\normtwov_i-\normtwov_{i'}}^2_2 -4n(d(d-1)+t(\theta+n+n_d-d-1))-6n
        \tag{$d\ge2$}\\
        &\le  \mu\cdot\min_{i\in[d],i'\in[d],\ i\ne i'}
        \norm{\normtwov_i-\normtwov_{i'}}^2_2 -4n(d(d-1)+t(\theta+n+n_d-d-1)) \\
        &\qquad -(n-1)
        (\min_{i\in[d],i'\in[d],\ i\ne i'}
        \norm{\normtwov_i-\normtwov_{i'}}_2+1) 
        (\min_{i\in[d],i'\in[d],\ i\ne i'}
        \norm{\normtwov_i-\normtwov_{i'}}_2)
        \tag{$\min_{i\in[d],i'\in[d],\ i\ne i'}
        \norm{\normtwov_i-\normtwov_{i'}}_2\le2$, using the triangle inequality}\\
        &\le  \mu\Biggl(
        \sqrt{\frac{\mu}{\mu+n-1}}\cdot
        \min_{i\in[d],i'\in[d],\ i\ne i'}
        \norm{\normtwov_i-\normtwov_{i'}}_2 -
        \left(1-\sqrt{\frac{\mu}{\mu+n-1}}\right)
        \Biggr)^2 \\
        &\qquad - 4n(d(d-1)+t(\theta+n+n_d-d-1)).
        \tag{\Cref{lemma: function lower bound}}\\
        &\le  \mu\cdot\min_{i\in[d],i'\in[d],\ i\ne i'}
        \norm{\normtwov_i(\kb)-\normtwov_{i'}(\kb)}^2_2 -4n(d(d-1)+t(\theta+n+n_d-d-1)).
        \tag{\Cref{lemma: copies lemma goods}}
    \end{align*}

    Combining the above, we get that an envy-free allocation exists.
\end{proof}

As an immediate corollary, we can get the following, simpler bound for the case of $d=n$ (i.e., if there is exactly one agent in each group).

\begin{corollary}\label{corollary: bound for d = n}
    Assume $d = n \ge2$ and $\min_{i,i'\in[d],\,i\ne i'} \norm{\normtwov_i-\normtwov_{i'}}_2^2>0$. If \[ k_z \ge \frac{4n^3}{\min_{i \in [d], i' \in [d],\ i\ne i'}\norm{\normtwov_i - \normtwov_{i'}}^2_2}\] for all $z \in [t]$ then an envy-free allocation exists.
\end{corollary}

\Cref{application: goods copies} implies that there exists a $\mu \in O( t n^3 / (g \eta) )$, where $\min_{i,i' \in [d],\ i\ne i'}\norm{\normtwov_i - \normtwov_{i'}}^2_2 \geq \eta$, such that if $k_z \geq \mu$, and $k_z \equiv 0 \pmod{g}$, for all $z \in [t]$, then an envy-free allocation exists. \cite{gorantla2023fair} prove that, for $d=2$, $\mu \in O(n^2 \sqrt{t}/(g \delta))$ and $\mu \in \Omega( n^2 / (g \delta) )$, when the two valuation functions are $\delta$-far from being identical; this means that the angle $\rho$ between the two valuations is at least $\delta$. For $t=2$,~\cite{gorantla2023fair} prove that $\mu \in O(n^2 /(g \delta))$. Using the fact that $\norm{\normtwov_i - \normtwov_{i'}}^2_2 = 2-2\cos\rho
\ge  2-2\cos \delta \ge \frac{4}{\pi^2} \delta^2$ we can get that our results imply that $\mu \in O(t n^3 / (g \delta^2) )$ for arbitrary $d, t$.  Interestingly, in~\Cref{corollary: bound for d = n} we see that $\mu \in O(n^3/ \eta)$ (or $\mu \in O(n^3/ \delta^2)$ using the $\delta$-far notion of~\cite{gorantla2023fair}), so there is no dependence on the number of item types $t$.

The lower bound of $\mu\in\Omega(n^2/(g\delta))$ for two groups leaves a gap from our general $O(tn^3/(g\delta^2))$ guarantee. Thus, our dependence on $t$, $n$, and the separation parameter is not tight; closing this gap, especially beyond the two-group or two-type settings, remains open.
    
  \section{Extensions}

In this section, we study various extensions of our technique to related problems in fair division.

\subsection{Envy-free Allocations of Chores} \label{sec: chores copies}

In this section, we study the problem of allocating a set of $m$ chores to $d$ groups of agents with additive preferences. 
Missing proofs can be found in \Cref{app: missing from EF chores}.

We define the \textsc{Log-Relative Norm} mechanism, which, given costs $c_1(\kb), \dots, c_d(\kb)$ outputs an allocation $x^{\choremech}$ such that, for every agent of group $i \in [d]$ and item $j \in \items$, we have
\[
x^{\choremech}_{i,j} = \frac{\ln \normc_{i,j}(\kb)}{n \ln \normc_{min}(\kb)}+\frac{1}{n}-\frac{ \ln \left( \prod_{i' \in [d]}\left(\normc_{i',j}(\kb) \right)^{n_{i'}} \right)}{n^2 \ln \normc_{min}(\kb)}.
\]
Notice that this is the same mechanism as the \textsc{Relative Norm} mechanism as defined in \Cref{sec: envy freeness}, but instead of $\normtwov$ we use $- \ln \normc$ as a proxy for an agent's disutility for a chore.
We start by showing that this mechanism outputs a feasible allocation.

\begin{lemma}
    The \textsc{Log-Relative Norm} mechanism is feasible: $x^{\choremech}_{i,j} \ge 0$ for all $i \in [d]$ and $j \in \items$, and $\sum_{i \in [d]} n_i \cdot x^{\choremech}_{i,j} = 1$ for all $j \in \items$.
\end{lemma}

\begin{proof}
    \begin{align*}
        x^{\choremech}_{i,j} & = \frac{\ln \normc_{i,j}(\kb)}{n \ln \normc_{min}(\kb)}+\frac{1}{n}-\frac{\ln \left( \prod_{i' \in [d]}\left(\normc_{i',j}(\kb) \right)^{n_{i'}} \right)}{n^2 \ln \normc_{min}(\kb)} \\
        & \\
        &\ge \frac{1}{n}-\frac{1}{n^2}\sum_{i' \in [d]} n_{i'} \frac{\ln \normc_{i',j}(\kb)}{\ln \normc_{min}(\kb)} \tag{$\frac{\ln \normc_{i,j}(\kb)}{n\ln \normc_{min}(\kb)}\ge 0$}\\
        &\ge \frac{1}{n}-\frac{\sum_{i' \in [d]} n_{i'}}{n^2} =0 \tag{$|\ln \normc_{i,j}(\kb)| \le |\ln \normc_{min}(\kb)|$}
    \end{align*}

    \begin{align*}
      \sum_{i\in[d]}n_i\cdot x^{\choremech}_{i,j} &=  \sum_{i\in[d]}n_i\cdot\Biggl(
      \frac{\ln\normc_{i,j}(\kb)}
      {n\ln\normc_{min}(\kb)}+\frac{1}{n} -\frac{\ln\left(
      \prod_{i'\in[d]}(\normc_{i',j}(\kb))^{n_{i'}}
      \right)}{n^2\ln\normc_{min}(\kb)}\Biggr)
      \\
      &=  \frac{\sum_{i\in[d]}n_i}{n}
      +\sum_{i\in[d]}n_i\Biggl(
      \frac{\sum_{i'\in[d]}n_{i'}
      \ln(\normc_{i,j}(\kb))}
      {n^2\ln\normc_{min}(\kb)} -\frac{\sum_{i'\in[d]}n_{i'}
      \ln(\normc_{i',j}(\kb))}
      {n^2\ln\normc_{min}(\kb)}\Biggr) \\
      &=1+\frac{1}{n^2\ln\normc_{min}(\kb)} \cdot \underbrace{\begin{aligned}[t]
		      &\sum_{i\in[d]}\sum_{i'\in[d]}
			    n_{i'}n_i\Bigl(\ln\normc_{i,j}(\kb)-
	      \ln\normc_{i',j}(\kb)\Bigr)
	      \end{aligned}}_{0} \\
    &=1.
    \end{align*}
\end{proof}

Recall that $\cgap_{i,i'}(x) = c_i(x_{i'}) - c_i(x_{i})$ is the advantage of an agent in group $i$ for an agent in group $i'$ in allocation $x$. The next lemma bounds $\cgap_{i,i'}(x^{\choremech})$ as a function of the KL-divergence between $\normc_i$ and $\normc_{i'}$.

\begin{lemma}\label{lemma: chores gap equility}
    For all $i, i' \in [d]$, $\cgap_{i,i'}(x^{\choremech}) = \frac{\norm{c_i(\kb)}_1}{- n \ln \normc_{min}(\kb)} D_{KL}(\normc_i(\kb) || \normc_{i'}(\kb))$.
\end{lemma}

\begin{proof}
    \begin{align*}
        \cgap_{i,i'}(x^{\choremech})
        &= \sum_{j\in\items}c_{i,j}(\kb)
        (x^{\choremech}_{i',j}-x^{\choremech}_{i,j})\\
        &= \frac{\norm{c_i(\kb)}_1}
        {n\ln\normc_{min}(\kb)}
        \sum_{j\in\items}\normc_{i,j}(\kb)
        \ln\left(
        \frac{\normc_{i',j}(\kb)}{\normc_{i,j}(\kb)}
        \right)\\
        &= \frac{\norm{c_i(\kb)}_1}
        {-n\ln\normc_{min}(\kb)}
        D_{KL}(\normc_i(\kb)||\normc_{i'}(\kb)).
    \end{align*}
\end{proof}

If $d=1$ and $k_z\equiv0\pmod n$ for every $z\in[t]$,
allocating $k_z/n$ copies of each type to each agent gives
an envy-free allocation. For the pairwise-gap LP and bounds
below, assume $d\ge2$.

Similarly to previous results in this paper, we next consider
a linear program for maximizing $\cgap$, this time for
$\ell_1$-normalized costs, and then argue that this LP has
a good integral solution.

\begin{equation}\label{LP_gap_chores}
\begin{array}{ll}
\text{maximize} & \alpha \\
\text{subject to}
 & \alpha\leq\displaystyle\sum_{j\in\items}
   \normc_{i,j}(\kb)(x_{i',j}-x_{i,j}) \forall i,i'\in[d],\ i'\neq i \\
 & \displaystyle\sum_{i=1}^{d}n_{i}\cdot x_{i,j}=1
   \quad\forall j\in\items \\
 & \alpha,x_{i,j}\geq0
   \quad\forall i\in[d],\ \forall j\in\items
\end{array}
\end{equation}

\begin{lemma}\label{lemma: rounding envy chores}
If $k_z \ge \theta$ and $k_z \equiv 0 \pmod{g}$ for all $z \in [t]$ then there exists an integral allocation $A$ such that, for all $i \in [d]$, $\min_{ i' \in [d],\ i'\ne i} \cgap_{i,i'}(A) \ge \norm{c_i(\kb)}_1 \cdot \alpha^* - 2(d(d-1) +t(\theta+n + n_d - d -1))\max_{j \in \items} c_{i,j}(\kb)$, where $\alpha^*$ is the optimal value of LP~\eqref{LP_gap_chores}.
\end{lemma}

The proof of the lemma is almost identical to that of \Cref{lemma: rounding envy goods} and is deferred to the appendix.

\begin{theorem} \label{theorem: Conditions for envy freeness chores}

    If $k_z \ge \theta$, $k_z \equiv 0 \pmod{g}$ for all $z \in [t]$, and \[ \normc_{max}(\kb) \leq  \frac{ \min_{\hat{i}, \tilde{i} \in [d],\ \hat{i}\ne\tilde{i}} D_{KL}(\normc_{\hat{i}}(\kb) || \normc_{\tilde{i}}(\kb))}{ 2n(d(d-1) +t(\theta+n + n_d - d -1)) \ln \normc_{min}(\kb)^{-1}},\] then an envy-free allocation exists.
\end{theorem}

\begin{proof}
Let $\alpha^*$ be the optimal value of LP~\eqref{LP_gap_chores}; \Cref{lemma: rounding envy chores} implies that there exists an integral solution $A$ such that, for all $i \in [d]$ , $\min_{ i' \in [d],\ i'\ne i} \cgap_{i,i'}(A) \geq \norm{c_i(\kb)}_1 \cdot \alpha^* - 2(d(d-1) +t(\theta+n + n_d - d -1))\max_{j \in \items} c_{i,j}(\kb)$.
$x^{\choremech}$ is a feasible solution to LP~\eqref{LP_gap_chores}, therefore, $\alpha^* \geq \min_{\hat{i}, \tilde{i} \in [d],\ \hat{i}\ne\tilde{i}} \frac{\cgap_{\hat{i}, \tilde{i}}(x^{\choremech})}{\norm{c_{\hat{i}}(\kb)}_1}$, which, by~\Cref{lemma: chores gap equility}, is at least $\min_{\hat{i}, \tilde{i} \in [d],\ \hat{i}\ne\tilde{i}} \frac{D_{KL}(\normc_{\hat{i}}(\kb) || \normc_{\tilde{i}}(\kb))}{- n \ln \normc_{min}(\kb)} $. Putting everything together, we have that for all $i \in [d]$, $\min_{i' \in [d],\ i'\ne i} \cgap_{i,i'}(A) \ge \norm{c_i(\kb)}_1\min_{\hat{i}, \tilde{i} \in [d],\ \hat{i}\ne\tilde{i}} \frac{D_{KL}(\normc_{\hat{i}}(\kb) || \normc_{\tilde{i}}(\kb))}{- n \ln \normc_{min}(\kb)} -  2(d(d-1) +t(\theta+n + n_d - d -1)) \max_{j \in \items} c_{i,j}(\kb)$. Therefore, if, for all $i \in [d]$,
\[
\max_{j \in \items} \normc_{i,j}(\kb) \leq
\min_{\substack{\hat{i}, \tilde{i} \in [d]\\ \hat{i}\ne\tilde{i}}}
\frac{D_{KL}(\normc_{\hat{i}}(\kb) || \normc_{\tilde{i}}(\kb))}
{2n\bigl(d(d-1)+t(\theta+n+n_d-d-1)\bigr)
\ln \normc_{min}(\kb)^{-1}},
\]
then $A$ is envy-free (since all $\cgap_{i,i'}(A)$ are nonnegative).
\end{proof}

\begin{corollary}
    If $d=n$ and  $\normc_{max}(\kb) \leq \min_{\hat{i}, \tilde{i} \in [d],\ \hat{i}\ne\tilde{i}} \frac{D_{KL}(\normc_{\hat{i}}(\kb) || \normc_{\tilde{i}}(\kb))}{ 2n^3 \ln \normc_{min}(\kb)^{-1}}$, then an envy-free allocation exists.
\end{corollary}

Before proving the main result of this section, we prove two useful technical lemmas.

\begin{lemma}\label{lemma: copies lemma chores}
For any pair $i, i' \in [d]$,
\[
D_{KL}(\normc_i(\kb) || \normc_{i'}(\kb))
\ge \frac{\alpha}{\beta} D_{KL}(\normc_i || \normc_{i'}) - \left(\ln \left( \frac{\beta}{\alpha} \right)
+ \frac{\beta-\alpha}{\alpha}
\ln \left(\normc_{min}\right)^{-1} \right),
\]
where $0<\alpha\le\beta$ and $k_z\in[\alpha,\beta]$ for each $z\in[t]$.
\end{lemma}

\begin{lemma}\label{lemma:technical inequality chores}
\[
x\left(\frac{x}{x+a}b- \ln \frac{x+a}{x} \right)
-c\ln(x+a) \ge \frac{b}{2}x
- \left(a\left(\frac{3}{2}b+1\right)
+ c\left(\ln\left(\frac{2c}{b}\right) - 1 \right) \right)
\]
where $a\ge0$ and $x,b,c>0$.
\end{lemma}

We are now ready to prove the main theorem for this section: a condition on the number of copies needed from each item to guarantee that an envy-free allocation for chores exists. The proof follows the same structure as the proof of \Cref{application: goods copies} and is deferred to the appendix.

\begin{theorem}\label{application: chores copies}
Assume
$\min_{\hat{i},\tilde{i}\in[d],\ \hat{i}\ne\tilde{i}}
D_{KL}(\normc_{\hat{i}}||\normc_{\tilde{i}})>0$, and let 
\[
\xi = n+ \lambda\left(
 \ln\left(
 \frac{2\lambda}{
 \min_{\hat{i}, \tilde{i} \in [d],\ \hat{i}\ne\tilde{i}}
 D_{KL}(\normc_{\hat{i}} || \normc_{\tilde{i}})}
 \right)-1 \right) + \left( \frac{5}{2}n+\lambda-1 \right)
 \ln \left(\normc_{min}\right)^{-1}.
 \]
 Then, if $k_z \ge \frac{2\xi}{\min_{\hat{i}, \tilde{i} \in [d],\ \hat{i}\ne\tilde{i}} D_{KL}(\normc_{\hat{i}} || \normc_{\tilde{i}})}$, and $k_z \equiv 0 \pmod{g}$ for all $z \in [t]$, an envy-free allocation exists, where $\lambda = 2n(d(d-1) +t(\theta+n + n_d - d -1))$.
\end{theorem}

\Cref{application: chores copies} implies that, for $\min_{\hat{i}, \tilde{i} \in [d],\ \hat{i}\ne\tilde{i}} D_{KL}(\normc_{\hat{i}} || \normc_{\tilde{i}}) \geq \eta$,
there exists a \[ \mu \in O\left( \frac{tn^3\left(\ln \left(\normc_{min}\right)^{-1} + \ln \left( \frac{tn^3}{g \eta}\right) \right) }{g \eta} \right),\] such that, if $k_z \geq \mu$, and $k_z \equiv 0 \pmod{g}$, for all $z \in [t]$, then an envy-free allocation exists.

\subsection{Envy-free Cake Cutting}\label{section: cake cutting}

In this section, we consider the classic problem of cake-cutting in the query model of Robertson and Webb. 

In this setting, each group has one agent ($n_1=\dots=n_d=1$) and every agent $i \in \agents$ has a valuation function $U_i$ that assigns a nonnegative value to any subinterval of $[0,1]$. We can think of these values as induced by a density $u_i$, i.e., for an interval $I = [a,b]$, $U_i(I) = \int_{x \in I} u_i(x) dx$. In the standard model, valuations are normalized, i.e., $U_i([0,1]) = \int_{x \in [0,1]} u_i(x) dx = 1$, additive, i.e., for a set of disjoint intervals $I_1, \dots, I_z$, $U_i (\cup_{j=1}^z I_j) = \sum_{j=1}^z U_i(I_j)$, and non-atomic, i.e., $\forall 0 \leq x \leq  y  \leq 1$, and $\lambda \in [0,1], \exists z \in [x, y]$ such that $U_i([x, z]) = \lambda \cdot U_i([x,y])$.
We will additionally assume that (i) for all $i \in \agents$, $u_i$ is $k$-Lipschitz continuous, and (ii) for all distinct $i, i' \in \agents$, $\norm{u_i - u_{i'}}_2 \geq \sqrt{\delta}$, where $\norm{u_i}_2 = \sqrt{ \int_{x \in [0,1]} u_i(x)^2 dx }$ and $\delta$ is a positive constant.

In the Robertson and Webb~\cite{robertson1998cake} model (henceforth, RW) there are two types of queries: (i) $\textsc{Eval}_i(I)$: given an agent $i$ and an interval $I \subseteq [0,1]$, return $U_i(I)$, and (ii) $\textsc{Cut}_i(x, z)$: given an agent $i$, a point $x \in [0,1]$, and a value $z\in[0,U_i([x,1])]$, return the smallest $y \in [x, 1]$ such that $U_i([x, y]) = z$.

It will be convenient to work with $\normtwou_i(x) = \frac{u_i(x)}{\norm{u_i}_2}$. By definition, $\int_{x \in [0,1]} \normtwou_i(x)^2 dx = 1$. Furthermore, $\normtwou_i$ is $\frac{k}{\norm{u_i}_2}$-Lipschitz continuous; since $\norm{u_i}_2 \geq 1$ (by the Cauchy-Schwarz inequality), $\normtwou_i$ is $k$-Lipschitz continuous, for all $i \in \agents$.
We derive below a corresponding lower bound for the distance between the normalized densities.

\begin{lemma}\label{lemma: Cake upper bound}
For all $i \in \agents$, $\normtwou_i(x) \le \max\{ (8k)^{1/3}, 2 \}$, for all $x \in [0,1]$.
\end{lemma}

\begin{proof}
Let $z^* = \max_{x \in [0,1]} \normtwou_i(x)$, and let $x^*$ be such that $\normtwou_i(x^*) = z^*$. Then, $\normtwou_i(x) \geq z^* - k | x - x^*|$. Wlog, $x^* \leq 1/2$; let $I = [x^*, x^* + 1/2]$. For all $x \in I$, it holds that $\normtwou_i(x) \geq z^* - \frac{k}{2}$. If $z^* \leq k/2$, then $z^* \leq k/2+\sqrt{2}$ is immediate. Otherwise, $z^*-k/2>0$, so the following estimate is valid: $1 = \int_{x \in [0,1]} \normtwou_i(x)^2 dx \geq \int_{x \in I} (z^* - \frac{k}{2})^2 dx = \frac{1}{2} (z^* - \frac{k}{2})^2$. Thus, in either case,
$z^* \leq \frac{k}{2} + \sqrt{2}$.
For $k \in [1,2\sqrt{2}]$, it is easy to confirm that $\frac{k}{2} + \sqrt{2}  \leq (8k)^{1/3}$, and for $k \leq 1$, we have that $z^* \leq \frac{1}{2} + \sqrt{2} < 2$.

Otherwise, if $k > 2\sqrt{2}$, we have that $\frac{z^*}{2k} \leq \frac{1}{4} + \frac{\sqrt{2}}{2k} \leq \frac{1}{2}$. Then, $x^* + \frac{z^*}{2k} \leq 1$. Let $I' = [x^*, x^* + \frac{z^*}{2k}]$. Repeating the same argument as above, we have that for all $x \in I'$, $\normtwou_i(x) \geq z^* - k \frac{z^*}{2k} = \frac{z^*}{2}$. We have $1 = \int_{x \in [0,1]} \normtwou_i(x)^2 dx \geq \int_{x \in I'} (\frac{z^*}{2})^2 dx = \frac{z^*}{2k} \cdot \frac{(z^*)^2}{4} = \frac{(z^*)^3}{8k}$. Re-arranging, we have that $z^* \leq (8k)^{1/3}$.
\end{proof}

For every $i\in\agents$, Lipschitz continuity and $\int_0^1u_i(x)\,dx=1$ imply $u_i(x)\le 1+k$ for all $x\in[0,1]$, and hence $\norm{u_i}_2^2\le 1+k$. Moreover, for every pair $i,i'\in\agents$,
\[
\left|\norm{u_i}_2-\norm{u_{i'}}_2\right|
\le \norm{u_i}_2\norm{u_{i'}}_2
\norm{\normtwou_i-\normtwou_{i'}}_2,
\]
since $\int_0^1\normtwou_i(x)\,dx=1/\norm{u_i}_2$ and
$\int_0^1\normtwou_{i'}(x)\,dx=1/\norm{u_{i'}}_2$, by the
Cauchy--Schwarz inequality. Using
\[
\norm{u_i-u_{i'}}_2^2
=\left(\norm{u_i}_2-\norm{u_{i'}}_2\right)^2
+\norm{u_i}_2\norm{u_{i'}}_2
\norm{\normtwou_i-\normtwou_{i'}}_2^2,
\]
we obtain, for all distinct $i,i'\in\agents$,
\begin{equation}\label{eq: cake normalized separation}
\norm{\normtwou_i-\normtwou_{i'}}_2^2
\ge \frac{\delta}{(1+k)(2+k)}.
\end{equation}

We are ready to prove our main result for this section.

\begin{theorem}\label{thm: cake cutting application}
If (i) for all $i \in \agents$, $u_i$ is $k$-Lipschitz continuous, and (ii) for all distinct $i, i' \in \agents$, $\norm{u_i - u_{i'}}_2 \geq \sqrt{\delta}$, then we can find a strongly envy-free allocation using at most
\[
n\!\left(1+\left\lceil
\max\left\{ k, \frac{(1+k)(2+k)}{\delta} \cdot\left(16n^3(\max\{(8k)^{1/3},2\})^2+\frac{7k}{2}\right)
\right\}
\right\rceil\right)
\]
\textsc{Eval} queries (and no \textsc{Cut} queries), in the RW model.
 \end{theorem}

 \begin{proof}
Choose $\epsilon$ such that
\[
\frac{1}{\epsilon}
=1+\left\lceil
\max\left\{ k, \frac{(1+k)(2+k)}{\delta} \cdot\left(16n^3(\max\{(8k)^{1/3},2\})^2+\frac{7k}{2}\right)
\right\}
\right\rceil.
\]
     First, we can cut the cake into $\frac{1}{\epsilon}$ pieces, each of length $\epsilon$.
     We can find the exact value of each agent for each piece using $n / \epsilon$ \textsc{Eval} queries; we will prove that, under the specified condition on the valuations, there is a strongly envy-free allocation of the pieces.                           

    We have established a lower bound on $\int_{0}^{1}(\normtwou_i(x)-\normtwou_{i'}(x))^2dx$, the squared $\ell_2$ distance between any pair of agents (recall that $\norm{\normtwou_i - \normtwou_{i'}}_2^2 \geq \frac{\delta}{(1+k)(2+k)}$). We would like to bound the difference of $\int_{0}^{1}(\normtwou_i(x)-\normtwou_{i'}(x))^2dx$ with its ``discretized'' version (w.r.t. the pieces we created). Specifically, we have that $\int_{\epsilon j}^{\epsilon(1+j)}\normtwou_{i}(x)dx$ is the real ($\ell_2$-normalized) value of agent $i$ for piece $j$. Treating pieces as indivisible items, $\sqrt{\sum_{j'=0}^{1/\epsilon-1} \left(\int_{\epsilon j'}^{ \epsilon ( j'+1)}\normtwou_i(x)dx\right)^2}$ is the $\ell_2$ norm of the ($\ell_2$-normalized) valuation vector of agent $i$; importantly, this is no longer equal to $1$, and therefore we must normalize again! Let $\frac{\int_{\epsilon j}^{\epsilon(1+j)}\normtwou_{i}(x)dx}{\sqrt{\sum_{j'=0}^{1/\epsilon-1} \left(\int_{\epsilon j'}^{ \epsilon ( j'+1)}\normtwou_i(x)dx\right)^2}}$ be the normalized value of agent $i$ for ``item'' $j$. Finally, let
    \[
    G_{i,i'}= \sum_{j=0}^{\frac{1}{\epsilon}-1}\Biggl(
    \frac{\int_{\epsilon j}^{\epsilon(1+j)}\normtwou_i(x)\,dx}
    {\sqrt{\sum_{j'=0}^{1/\epsilon-1}
    \left(\int_{\epsilon j'}^{\epsilon(j'+1)}\normtwou_i(x)\,dx\right)^2}}
    -
    \frac{\int_{\epsilon j}^{\epsilon(1+j)}\normtwou_{i'}(x)\,dx}
    {\sqrt{\sum_{j'=0}^{1/\epsilon-1}
    \left(\int_{\epsilon j'}^{\epsilon(j'+1)}\normtwou_{i'}(x)\,dx\right)^2}}
    \Biggr)^2
    \]
    be the squared $\ell_2$ distance between two valuation vectors with respect to the \emph{indivisible} items/pieces created. 
    
    To apply~\Cref{corollary: Conditions for envy freeness with single agents groups} we need a bound on $\min_{i, i' \in \agents,\ i\ne i'} G_{i,i'}$. The following lemma, whose proof is deferred to~\Cref{app: missing from cake cutting}, gives us our desired bound.

\begin{lemma} \label{lemma: Discrete cake eucledian bound}
For all $\epsilon>0$ such that $1/\epsilon\in\mathbb{N}$ and $k\epsilon\leq1$, and $i, i' \in \agents$,
    \[
    G_{i,i'} \geq \int_{0}^{1}(\normtwou_i(x)-\normtwou_{i'}(x))^2dx - \frac{7 k \epsilon}{2}.
    \]
Moreover, for every $i\in\agents$,
\[
\sqrt{\sum_{j=0}^{1/\epsilon-1}
\left(\int_{\epsilon j}^{\epsilon(j+1)}
\normtwou_i(x)\,dx\right)^2}
\geq
\sqrt{\epsilon}\left(1-\frac{k\epsilon}{2}\right).
\]
\end{lemma}

Combining~\Cref{lemma: Cake upper bound,lemma: Discrete cake eucledian bound}, when $k\epsilon<1$ we have, for every $i$ and $j$,
\[
\int_{\epsilon j}^{\epsilon(j+1)}\normtwou_i(x)\,dx
\le \epsilon\max\{(8k)^{1/3},2\},
\]
and
\[
\sqrt{\sum_{\ell=0}^{1/\epsilon-1}
\left(\int_{\epsilon\ell}^{\epsilon(\ell+1)}
\normtwou_i(x)\,dx\right)^2}
\ge \sqrt{\epsilon}\left(1-\frac{k\epsilon}{2}\right).
\]
Thus, the $j$-th coordinate of the $\ell_2$-normalized discrete valuation vector satisfies
\begin{align*}
\frac{\int_{\epsilon j}^{\epsilon(j+1)}\normtwou_i(x)\,dx}
{\sqrt{\sum_{\ell=0}^{1/\epsilon-1}
\left(\int_{\epsilon\ell}^{\epsilon(\ell+1)}
\normtwou_i(x)\,dx\right)^2}} &\le
\frac{\epsilon\max\{(8k)^{1/3},2\}}
{\sqrt{\epsilon}(1-k\epsilon/2)}\\
&=
\frac{\sqrt{\epsilon}\max\{(8k)^{1/3},2\}}
{1-k\epsilon/2}\\
&< 2\sqrt{\epsilon}\max\{(8k)^{1/3},2\}.
\end{align*}

\Cref{lemma: Discrete cake eucledian bound}, together with~\eqref{eq: cake normalized separation}, implies that
\[
\sqrt{\min_{i, i' \in \agents,\ i\ne i'}G_{i,i'}}
\geq
\sqrt{\frac{\delta}{(1+k)(2+k)}-\frac{7k\epsilon}{2}}.
\]
Therefore, the condition in~\Cref{corollary: Conditions for envy freeness with single agents groups} holds strictly if
\[
2\sqrt{\epsilon}\max\{(8k)^{1/3},2\}
<
\frac{\sqrt{\frac{\delta}{(1+k)(2+k)}-\frac{7k\epsilon}{2}}}{2n\sqrt n},
\]
or, equivalently, if
\[
\epsilon<
\frac{\delta}{(1+k)(2+k)
\left(16n^3(\max\{(8k)^{1/3},2\})^2+\frac{7k}{2}\right)}.
\]
This is guaranteed by our choice of $\epsilon$. The strict inequality makes every resulting envy gap positive, so the allocation is strongly envy-free.
 \end{proof}

By treating $k$ and $\delta$ as constants, we obtain the following corollary, which provides a more intuitive bound on the query complexity.

\begin{corollary}[of~\Cref{thm: cake cutting application}]\label{cor: cake}
If (i) for all $i \in \agents$, $u_i$ is $k$-Lipschitz continuous, and (ii) for all distinct $i, i' \in \agents$, $\norm{u_i - u_{i'}}_2 \geq \sqrt{\delta}$, for some constants $k$ and $\delta$, then we can find a strongly envy-free allocation in the RW model using $O(n^4)$ \textsc{Eval} queries (and no \textsc{Cut} queries).
\end{corollary} 
\subsection{Proportional Allocations of Goods and Chores}\label{sec: proportionality}

In this section, we study conditions under which proportional allocations exist for goods and chores. We also give applications of our existence results.

We will no longer assume that agents are organized into groups. This relaxation does not make the setting any more restrictive, as our bounds will not depend on the pairwise distances between agents. Consequently, the results remain meaningful even when different agents have identical or nearly identical valuations. Moreover, since proportionality is a significantly weaker notion of fairness than envy-freeness, we will not require that agents with identical valuations receive identical bundles. Treating each agent as an individual group therefore simplifies our exposition.

Throughout this section, we also assume that $\kb = \one$. While it is possible to replicate the analysis and derive slightly stronger bounds on the number of copies required to guarantee proportionality (as opposed to envy-freeness), such an analysis would closely mirror that of the previous sections and would not yield additional technical insights. Instead, our focus is on exploring how the connection between fair allocations and $f$-divergences can be leveraged to obtain results by drawing on the heavy machinery developed in information theory for these divergences. In particular, we employ the variational characterization of $f$-divergences to establish bounds that we then use to recover, up to constant factors, state-of-the-art results in stochastic fair division.
All missing proofs are deferred to~\Cref{app: missing from proportionality}.

\subsubsection{Allocating Goods via the Trading Post Mechanism.}

We define $S_j=\frac{1}{n}\sum_{i \in \agents} \frac{v_{i,j}}{v_i(\items)}=\frac{1}{n}\sum_{i \in \agents} \normv_{i,j}$ as society's valuation for item $j \in \items$. Equivalently, $S_j$ is the average of agents' $\ell_1$-normalized valuations. Items with $S_j=0$ are universally worthless, so they may be ignored during the analysis and allocated arbitrarily afterward. Hence, we restrict attention to items with $S_j>0$. For every such item, the \textsc{Trading Post} mechanism outputs an allocation $x^{\trading}$ such that $x^{\trading}_{i,j} = \frac{\normv_{i,j}}{\sum_{i' \in \agents} \normv_{i',j}}$, for all $i \in \agents$ and $j \in \items$.

We start by observing that the utility of an agent in the Trading Post mechanism has a nice form, as a function of the $\chi^2$-divergence between her valuation and society's valuation.

\begin{lemma}
\label{lemma: utility under TP}
    For all $i \in \agents$, $\frac{v_i(x^{\trading}_i)}{v_i(\items)} = \frac{1}{n}\left(1+ D_{\chi ^2}(\normv_i || S) \right)$. 
\end{lemma}

\begin{proof}

\begin{align*}
    \frac{v_i(x^{\trading}_i)}{v_i(\items)} &= \frac{1}{{v_i(\items)}} \sum_{j \in [m]} v_{i,j} \frac{\normv_{i,j}}{\sum_{i' \in \agents} \normv_{i',j}} \\
    &= \sum_{j \in [m]} \frac{(\normv_{i,j})^2}{\sum_{i' \in \agents} \normv_{i',j}} \\
    &= \frac{1}{n}\left(\sum_{j \in [m]}\frac{(\normv_{i,j})^2}{\frac{1}{n}\sum_{i' \in \agents} \normv_{i',j}}\right)\\
    &= \frac{1}{n}\left(
    \begin{aligned}
    1 + \sum_{j \in [m]}\biggl(&\frac{(\normv_{i,j})^2}{S_j} -2\normv_{i,j}+S_j\biggr)
    \end{aligned}
    \right) \tag{$\norm{\normv_i}_1 = \norm{S}_1 =1$}\\
    &= \frac{1}{n}\left(1 +\sum_{j \in [m]} \frac{(\normv_{i,j})^2 - 2\normv_{i,j}S_j +S^2_j}{S_j}  \right) \\
    &= \frac{1}{n}\left(1 +\sum_{j \in [m]} \frac{(\normv_{i,j}-S_j)^2}{S_j}  \right)\\
    &= \frac{1}{n}\left(1 + D_{\chi ^2}(\normv_i || S)\right).
\end{align*}
\end{proof}

Now, consider the following LP that maximizes the $\ell_1$-normalized minimum utility of an agent. An allocation is proportional if its value with respect to this LP is at least $1/n$. 

\begin{equation}\label{LP1}
\begin{array}{ll@{}ll}
\text{maximize}  & \alpha &\\
\text{subject to}& \alpha \le \sum_{j=1}^{m} \normv_{i,j} x_{i,j} \hspace{5mm}& \forall i \in \agents\\
& \sum_{i=1}^{n} x_{i,j} \leq 1  &\forall j \in \items\\
& \alpha, x_{i,j} \geq 0 & \forall i \in \agents, \forall j \in \items
\end{array}
\end{equation}

The next lemma allows us to argue about the existence of an integral allocation that has high objective value (with respect to the LP above).

\begin{lemma}\label{lemma: rounding prop allocation}[Essentially Theorem 3 of~\cite{eckart2024fairness}]
    Let $x^*$ be an optimal solution to LP~\eqref{LP1}. There exists an integral allocation $A$, such that, for all agents $i \in \agents$, $v_i(A_i) \geq v_i(x^*_i) - \max_{j \in \items} v_{i,j}$. 
\end{lemma}

\begin{proof}
     Let $y$ maximize $\sum_{i\in\agents}v_i(y_i)$ over feasible fractional allocations satisfying $v_i(y_i)\ge v_i(x_i^*)$ for every $i\in\agents$. By LP duality, $y$ maximizes $\sum_{i\in\agents}w_i v_i(y_i)$ for strictly positive weights $w_i=1+\lambda_i$, where $\lambda_i\ge0$ are the dual multipliers of these utility lower bounds. Setting $p_j=\max_{i\in\agents}w_i v_{i,j}$, and  $b_i=\sum_{j\in\items}p_jy_{i,j}$, complementary slackness implies that $(y,p,b)$ is a Fisher market equilibrium.
    
     Applying the rounding algorithm of~\cite{barman2019proximity}, with the utility bound established in~\cite{eckart2024fairness}, gives an integral allocation $A$ such that, for every $i\in\agents$,
\[
v_i(A_i)
\ge v_i(y_i)-\max_{j\in\items}v_{i,j}
\ge v_i(x_i^*)-\max_{j\in\items}v_{i,j}.
\]
The lemma follows.
\end{proof}

Given~\Cref{lemma: utility under TP,lemma: rounding prop allocation}, we are ready to prove our main result in this section.

\begin{theorem} \label{theorem: prop condition}
    Let $D^* = \min_{i \in \agents} \left\{ D_{\chi^2}(\normv_i || S) \right\}$. If $\normv_{max} \le \frac{D^*}{n}$, then a proportional allocation exists. When the inequality is strict, then a strongly proportional allocation exists.
\end{theorem}

\begin{proof}
    Let $\alpha^*$ be the value of an optimal solution to LP~\eqref{LP1}, and let $A$ be the integral allocation from~\Cref{lemma: rounding prop allocation}.
    We have that $v_i(A_i) \ge \alpha^* \, v_i(\items) - \max_{j \in \items} v_{i,j}$. Let $x^{\trading}$ be the allocation output by the Trading Post mechanism. Since $x^{\trading}$ is a feasible solution for LP~\eqref{LP1}, we have that $\alpha^* \geq \min_{i \in \agents} \normv_i(x^{\trading}_i) = \min_{i \in \agents} \frac{v_i(x^{\trading}_i)}{v_i(\items)}$. Therefore, using~\Cref{lemma: utility under TP}, we have that $\alpha^* \ge \frac{1}{n}\left(1+  D^*\right)$, where $D^* = \min_{i \in \agents} \left\{ D_{\chi ^2}(\normv_i||S) \right\}$. Therefore, we overall have $v_i(A_i) \geq \frac{1}{n}\left(1+  D^*\right) v_i(\items)- \max_{j \in \items} v_{i,j} \ge v_i(\items)/n$, where in the last inequality we used the fact that, for all $i \in \agents$ and $j \in \items$, $\frac{v_{i,j}}{v_{i}(\items)} \le \frac{D^*}{n}$.
\end{proof}

\subsubsection{Proportional Allocations of Goods in the Stochastic Setting.}

As an application of~\Cref{theorem: prop condition}, we consider a setting with stochastic valuations. Specifically, we consider the setting where $v_{i,j}$ is sampled i.i.d. from the $U[0,1]$ distribution. \cite{manurangsi2021closing} prove that a proportional allocation exists with high probability for $m \geq n$ (their result, in fact, holds as long as the PDF of the distribution is upper and lower bounded by constants). Here, we re-prove this result in a ``modular'' way, but with a worse dependence on $n$. Specifically, we show that in this model $\normv_{max} \le \frac{\min_{i \in \agents} \left\{ D_{\chi^2}(\normv_i || S) \right\}}{n}$ holds with high probability, for $m \in \Omega(n)$. This readily implies the existence of a proportional allocation.

\begin{lemma}\label{lemma: stochastic prop}
    In the i.i.d. $U[0,1]$ model $(v_{i,j} \sim U[0,1])$, for $m \in \Omega(n)$, it holds that $D_{\chi^2}(\normv_i || S) \ge \frac{n-28}{39n}$, for all $i \in \agents$, with high probability $(1-O(1/n))$.
\end{lemma}

\begin{theorem}\label{thm: iid goods}
    In the i.i.d. $U[0,1]$ model $(v_{i,j} \sim U[0,1])$, a (strongly) proportional allocation exists with high probability $(1-O(1/n))$, when $m \in \Omega(n)$.
\end{theorem}

\begin{proof}
    \Cref{lemma: stochastic prop} implies that, with high probability, $D^* = \min_{i \in \agents} \left\{ D_{\chi ^2}(\normv_i || S) \right\} \ge \frac{n-28}{39n}$. Also with high probability, for all $i \in \agents, j \in \items$, $\frac{v_{i,j}}{v_i(\items)} \le \frac{1}{v_i(\items)} \leq \frac{1}{m}\left(\frac{1}{1/2 - \sqrt{\frac{\log n}{m}}} \right)$. \Cref{theorem: prop condition} implies that if $\frac{1}{m}\left(\frac{1}{1/2 - \sqrt{\frac{\log n}{m}}} \right) \le \frac{D^*}{n}$ then a proportional allocation exists. Therefore, if $\frac{1}{m}\left(\frac{1}{1/2 - \sqrt{\frac{\log n}{m}}}\right) \le \frac{n-28}{39n^2}$ a proportional allocation exists; assuming that $m \ge c' n$, for some positive constant $c'$, the strict inequality holds for all sufficiently large $n$.
\end{proof} 

\subsubsection{Allocating Chores via the Inverse Trading Post Mechanism.}\label{sec: allocate chores prop}

Let $H_j=\frac{n}{\sum_{i \in \agents} 1/\normc_{i,j}}$ be society's cost for item $j \in \items$, defined as the harmonic mean of normalized agent costs. We define a variation of the Trading Post mechanism where weights are inversely proportional to agents' costs. Formally, given costs $c_1, \dots, c_n$, the \textsc{inverse Trading Post} outputs an allocation $x^{\invtrading}_{i,j} = \frac{1/\normc_{i,j}}{\sum_{i' \in \agents} 1/\normc_{i',j}}$. The cost of agent $i \in \agents$ when participating in the inverse Trading Post mechanism is $c_i(x^{\invtrading}_i) = \sum_{j \in \items } c_{i,j} \frac{1/\normc_{i,j}}{\sum_{i' \in \agents} 1/\normc_{i',j}} = \sum_{j \in \items } \frac{c_i(\items)}{\sum_{i' \in \agents} 1/\normc_{i',j}} =  \frac{c_i(\items)}{n} \sum_{j \in \items } H_j$.
We start by expressing this cost as a function of the distance between an agent and society.

\begin{lemma}
\label{lemma: disutility under TP}
    For all $i \in \agents$, $\frac{c_i(x^{\invtrading}_i)}{c_i(\items)} = \frac{1}{n}\left( 1-\frac{1}{n} \sum_{\ell \in \agents} \sum_{j \in \items} \frac{(\normc_{\ell, j} - H_j)^2}{\normc_{\ell, j}}\right)$.
\end{lemma}

Interestingly, the expression for the normalized cost of agent $i$ in~\Cref{lemma: disutility under TP} does not depend on $i$'s identity, i.e., all agents have the same normalized cost.

Next, consider the following LP that minimizes the maximum $\ell_1$-normalized cost of an agent.

\begin{equation}\label{LP2}
\begin{array}{ll@{}ll}
\text{minimize}  & \alpha &\\
\text{subject to}& \alpha \ge \sum_{j=1}^{m}\normc_{ij}x_{ij} \hspace{5mm}& \forall i \in \agents \\
& \sum_{i=1}^{n} x_{ij} = 1  &\forall j \in \items \\
& x_{ij} \geq 0 & \forall i \in \agents, \forall j \in \items
\end{array}
\end{equation}

Similarly to~\Cref{lemma: rounding prop allocation}, the following lemma argues about the existence of a good integral allocation.

\begin{lemma}\label{lemma: rounding prop chores allocation}[Essentially Theorem 11 of~\cite{eckart2024fairness}]
    Let $x^*$ be an optimal solution to LP~\eqref{LP2}. There exists an integral allocation $A$, such that, for all agents $i \in \agents$, $c_i(A_i) \leq c_i(x^*_i) + \max_{j \in \items} c_{i,j}$. 
\end{lemma}

\begin{proof}
Let $\alpha^*$ be the optimal value of LP~\eqref{LP2}.
Take an optimal dual solution $(\lambda,r)$, where
$\lambda_i\ge0$ and $\sum_{i\in\agents}\lambda_i=1$, $r_j=\min_{i\in\agents}\lambda_i\normc_{i,j}$, and  $\sum_{j\in\items}r_j=\alpha^*$. Since all costs are strictly positive and every chore must be allocated, $\alpha^*>0$. If some $\lambda_i=0$, then every $r_j\le0$, a contradiction, and therefore all $\lambda_i$ and $r_j$ are strictly positive.
Complementary slackness gives
$x^*_{i,j}>0$ only if $r_j=\lambda_i\normc_{i,j}$, and
\[
e_i:=\sum_{j\in\items}r_jx^*_{i,j}
=\lambda_i\alpha^*>0.
\]
Hence each agent receives only chores maximizing reward per
unit cost and earns exactly $e_i$. Together with completeness,
this shows that $(x^*,r,e)$ is a market equilibrium. The rounding algorithm of~\cite{branzei2019algorithms}, with the cost bound established in~\cite{eckart2024fairness}, therefore gives an integral allocation $A$ satisfying $c_i(A_i)\le c_i(x_i^*)+\max_{j\in\items}c_{i,j}$ 
for every $i\in\agents$. The lemma follows.
\end{proof}

We are now ready to prove the main result of this section.

\begin{theorem}\label{theorem: condition for prop chores}
    If $\normc_{max} \le \frac{1}{n^2}\sum_{\ell \in \agents} \sum_{j \in \items} \frac{(\normc_{\ell, j} - H_j)^2}{\normc_{\ell, j}}$, then a proportional allocation exists. When the inequality is strict, then a strongly proportional allocation exists.
\end{theorem}

\begin{proof}
    The allocation of~\Cref{lemma: rounding prop chores allocation} satisfies, for all $i \in \agents$, $c_i(A_i) \leq c_i(x^*_i) + \max_{j \in \items} c_{i,j}$. Furthermore, $\frac{c_i(x^*_i)}{c_i(\items)} \leq \alpha^*$, where $\alpha^*$ is the optimal value of LP~\eqref{LP2}.
    The allocation of the inverse Trading Post mechanism, $x^{\invtrading}$, is a feasible solution for LP~\eqref{LP2}, therefore,  
    $\alpha^* \leq \max_{z \in \agents}\frac{c_{z}(x^{\invtrading}_{z})}{ c_{z}(\items) }$. 
    By~\Cref{lemma: disutility under TP}, the RHS is equal to $\frac{1}{n}\left(1 - \frac{1}{n} \sum_{\ell \in \agents} \sum_{j \in \items} \frac{(\normc_{\ell, j} - H_j)^2}{\normc_{\ell, j}} \right)$.
    Combining, we have that 
    \begin{align*}
    c_i(A_i) &\leq \frac{c_i(\items)}{n}\left(1 - \frac{1}{n} \sum_{\ell \in \agents} \sum_{j \in \items} \frac{(\normc_{\ell, j} - H_j)^2}{\normc_{\ell, j}} \right) + \max_{j \in \items} c_{i,j}\\
    &\leq \frac{c_i(\items)}{n}\left(1 - \frac{1}{n} \sum_{\ell \in \agents} \sum_{j \in \items} \frac{(\normc_{\ell, j} - H_j)^2}{\normc_{\ell, j}} \right) + \frac{c_i(\items)}{n^2}\sum_{\ell \in \agents} \sum_{j \in \items} \frac{(\normc_{\ell, j} - H_j)^2}{\normc_{\ell, j}}\\
    &= \frac{c_i(\items)}{n}.
    \end{align*}
\end{proof}

We note that there exists an equivalent version of this theorem, where the bound on $\normc_{max}$ involves the $\chi^2$-divergence between $\normc_i$ and $H/\norm{H}_1$. Specifically, if $$\normc_{max} \leq\frac{1}{n^2} \sum_{\ell \in \agents} \left(  (1-\norm{H}_1)^2 +  \norm{H}^2_1  D_{\chi ^2}(\frac{H}{\norm{H}_1} || \normc_{\ell})  \right),$$ then a proportional allocation exists.

\subsubsection{Proportional Allocations of Chores in the Stochastic Setting.}

As an application of~\Cref{theorem: condition for prop chores}, similarly to the case of goods, we consider a setting with stochastic preferences: $c_{i,j}$ is sampled i.i.d. from the $U[0,1]$ distribution. In recent work, \cite{manurangsi2025asymptotic} prove that, for all distributions whose PDF is bounded by constants, a proportional allocation exists with high probability for $m \in \omega(1)$. Here, we re-prove a weaker version of this result, by showing that the condition in~\Cref{theorem: condition for prop chores} holds with high probability for $m \in \Omega(n)$ (and therefore, by~\Cref{theorem: condition for prop chores}, a proportional allocation exists).

We start by showing a couple of technical lemmas.

\begin{lemma}
\label{lemma: chores variational charachterization}
    $\sum_{j \in \items} \frac{(\normc_{i, j} - H_j)^2}{\normc_{i, j}} \ge \langle H - \normc_i, g \rangle - \frac{\langle \normc_i, g^2 \rangle}{4}$, where $g$ is any vector in $\mathbb{R}^m$ and $g^2 = [g_1^2, g_2^2, \cdots, g_m^2]$. 
\end{lemma}

\begin{lemma} \label{lemma: stochastic chores}
    In the i.i.d. $U[0,1]$ model $(c_{i,j} \sim U[0,1])$, for $m \in \Omega(n)$, it holds that $\sum_{j \in \items} \frac{(\normc_{i, j} - H_j)^2}{\normc_{i, j}} \ge \frac{1}{16}$, for all $i \in \agents$, with high probability $(1-O(1/n))$. 
\end{lemma}

We are now ready to state our result about the existence of proportional allocations of chores.

\begin{theorem}\label{thm: chores iid proportional}
    In the i.i.d. $U[0,1]$ model $(c_{i,j} \sim U[0,1])$, a (strongly) proportional allocation exists with high probability $(1-O(1/n))$, when $m \in \Omega(n)$.
\end{theorem}

\begin{proof}
    \Cref{lemma: stochastic chores} gives that, with high probability, $\sum_{j \in \items} \frac{(\normc_{i, j} - H_j)^2}{\normc_{i, j}} \ge \frac{1}{16}$, for all $i \in \agents$. Therefore, with high probability, $\sum_{\ell \in \agents} \sum_{j \in \items} \frac{(\normc_{\ell, j} - H_j)^2}{\normc_{\ell, j}} \geq \frac{n}{16}$. \Cref{theorem: condition for prop chores} then implies that if, for all $i \in \agents$, $\frac{n^2}{c_i(\items)} \max_{j \in \items} c_{i,j}  \leq \frac{n}{16}$, or, equivalently, for all $i \in \agents, j \in \items$, $\frac{c_{i,j}}{c_i(\items)} \leq \frac{1}{16n}$, then a proportional allocation exists. 
    With high probability, we have that for all $i \in \agents, j \in \items$, $\frac{c_{i,j}}{c_i(\items)} \leq \frac{1}{c_i(\items)} \leq \frac{1}{m}\left(\frac{1}{1/2 - \sqrt{\frac{\log n}{m}}} \right)$. Therefore, for $m \geq c' n$, for some constant $c'$, the desired condition holds for sufficiently large $n$; the theorem follows.
\end{proof}   

\section*{Acknowledgements}

Alexandros Psomas and Marios Mertzanidis were supported in part by an NSF CAREER award CCF-2144208.

We thank the anonymous FOCS reviewers for their constructive feedback and suggestions. In preparing the camera-ready version of this work, we used OpenAI's Codex to help identify possible errors in proofs and complexity bounds, draft revisions to the exposition, and prepare the manuscript. The most substantial changes from earlier versions concern the normalization of valuation vectors in~\Cref{section: cake cutting} and the corresponding revisions to our cake-cutting analysis. The authors take full responsibility for all the content.

\appendix

\section{Proofs missing from~\Cref{sec: envy freeness}}\label{app: missing from main EF section}

\begin{proof}[Proof of~\Cref{lemma: copies lemma goods}]
First, observe that for every group $\ell \in [d]$:
\begin{align*}
    \norm{v_\ell(\kb)}_2 &= \sqrt{\sum_{j \in \items} (v_{\ell,j}(\kb))^2} \\
    & = \sqrt{\sum_{z \in [t]} k_z \cdot (v_{\ell,z})^2} \\
    & \ge \sqrt{\sum_{z \in [t]}\underline{k} \cdot (v_{\ell,z})^2} \\
    & = \sqrt{\underline{k}} \norm{v_\ell}_2.
\end{align*}
Similarly:
\begin{align*}
    \norm{v_\ell(\kb)}_2 &= \sqrt{\sum_{j \in \items} (v_{\ell,j}(\kb))^2} \\
    &= \sqrt{\sum_{z \in [t]} k_z \cdot (v_{\ell,z})^2} \\
    &\le \sqrt{\sum_{z \in [t]} \overline{k} \cdot (v_{\ell,z})^2} \\
    &= \sqrt{\overline{k}} \norm{v_\ell}_2.
\end{align*}

Thus there exist $\gamma_i$, $\gamma_{i'}$ such that $\norm{v_i(\kb)}_2 = \frac{\norm{v_i}_2}{\gamma_i}$  and $\norm{v_{i'}(\kb)}_2 = \frac{\norm{v_{i'}}_2}{\gamma_{i'}}$, such that $ \frac{1}{\sqrt{\overline{k}}}\le \gamma_i \le \frac{1}{\sqrt{\underline{k}}}$ and $ \frac{1}{\sqrt{\overline{k}}}\le \gamma_{i'} \le \frac{1}{\sqrt{\underline{k}}}$. Without loss of generality, assume that $\gamma_{i'} \ge \gamma_{i}$. We have that:

\begin{align*}
&\norm{\normtwov_i(\kb) - \normtwov_{i'}(\kb)}_2  = \sqrt{\sum_{j \in \items} (\normtwov_{i,j}(\kb) - \normtwov_{i',j}(\kb))^2} \\
&\quad = \sqrt{\sum_{j \in \items} \left(
\begin{aligned}
&\frac{v_{i,j}(\kb)}{\norm{v_i(\kb)}_2} - \frac{v_{i',j}(\kb)}{\norm{v_{i'}(\kb)}_2}
\end{aligned}
\right)^2} \\
&\quad = \sqrt{\sum_{j \in \items} \left(
\begin{aligned}
&\gamma_i\frac{v_{i,j}(\kb)}{\norm{v_i}_2} - \gamma_{i'}\frac{v_{i',j}(\kb)}{\norm{v_{i'}}_2}
\end{aligned}
\right)^2} \\
&\quad = \sqrt{\sum_{j \in \items} \left(
\begin{aligned}
&\gamma_i\left(\frac{v_{i,j}(\kb)}{\norm{v_i}_2} - \frac{v_{i',j}(\kb)}{\norm{v_{i'}}_2}\right) - (\gamma_{i'}-\gamma_{i})\frac{v_{i',j}(\kb)}{\norm{v_{i'}}_2}
\end{aligned}
\right)^2} \\
&\quad \ge \gamma_i\sqrt{\sum_{j \in \items}\left(\frac{v_{i,j}(\kb)}{\norm{v_i}_2} - \frac{v_{i',j}(\kb)}{\norm{v_{i'}}_2}\right)^2} - (\gamma_{i'}-\gamma_{i}) \sqrt{\sum_{j \in \items} \left(\frac{v_{i',j}(\kb)}{\norm{v_{i'}}_2} \right)^2}  \tag{Triangle Inequality}\\
&\quad = \gamma_i\sqrt{\sum_{z \in [t]} k_z\left(\frac{v_{i,z}}{\norm{v_i}_2} - \frac{v_{i',z}}{\norm{v_{i'}}_2}\right)^2}  - \frac{\gamma_{i'}-\gamma_{i}}{\gamma_{i'}} \underbrace{\sqrt{\sum_{j \in \items} \left(\frac{v_{i',j}(\kb)}{\norm{v_{i'}(\kb)}_2} \right)^2}}_{1}\\
&\quad \ge \gamma_i\sqrt{\sum_{z \in [t]} \underline{k}\left(\frac{v_{i,z}}{\norm{v_i}_2} - \frac{v_{i',z}}{\norm{v_{i'}}_2}\right)^2}  - \frac{\gamma_{i'}-\gamma_{i}}{\gamma_{i'}} \\
&\tag{$k_z \ge \underline{k}$}\\
&\quad = \gamma_i \sqrt{\underline{k}} \cdot \sqrt{\sum_{z \in [t]}\left(\frac{v_{i,z}}{\norm{v_i}_2} - \frac{v_{i',z}}{\norm{v_{i'}}_2}\right)^2}  - \frac{\gamma_{i'}-\gamma_{i}}{\gamma_{i'}}\\
&\quad = \gamma_i \cdot \sqrt{\underline{k}} \cdot \norm{\normtwov_i - \normtwov_{i'}}_2  - \frac{\gamma_{i'}-\gamma_{i}}{\gamma_{i'}} \\
&\quad \ge \sqrt{\frac{\underline{k}}{\overline{k}}} \cdot \norm{\normtwov_i - \normtwov_{i'}}_2  - \left(1-\sqrt{\frac{\underline{k}}{\overline{k}}}\right).
\end{align*}

\end{proof}

\begin{proof}[Proof of~\Cref{lemma: function lower bound}]
    \begin{align*}
        x\left( \sqrt{\frac{x}{x+a}} \cdot b - \left(1-\sqrt{\frac{x}{x+a}}\right)\right)^2 &= x\left( (1+a/x)^{-1/2} \cdot (b+1) - 1\right)^2 \\
        & \ge x\left(\left(1-\frac{a}{2x} \right) \cdot (b+1) -1 \right)^2 \\
        & = x\left(b -\frac{a(b+1)}{2x} \right)^2 \\
        & \ge x\left(b^2 -\frac{ab(b+1)}{x} \right) \\
        & = xb^2 -ab(b+1)
    \end{align*}
    where in the second line we used the fact that $(1+a/x)^{-1/2} \ge 1-\frac{a}{2x}$ and $x\ge \frac{a(b+1)}{2b}$.
\end{proof}

\section{Proofs missing from~\Cref{sec: chores copies}}\label{app: missing from EF chores}

\begin{proof}[Proof of~\Cref{lemma: rounding envy chores}]

    Choose an optimal basic feasible solution of LP~\eqref{LP_gap_chores}. LP~\eqref{LP_gap_chores} has $d \, m + 1$ variables ($\alpha$ and the $x_{i,j}$s), therefore, this solution has at least $d \, m + 1$ tight constraints. There are $d(d-1) + m +  d \, m + 1$ constraints, out of which $d \, m $ constraints are non-negativity constraints for $x_{i,j}$. By simple counting, we have that at least $d \, m + 1 - (d(d-1) + m + 1)$ out of these non-negativity constraints must be tight. Therefore, at most $m + d(d-1)$ of the ``$x_{i,j} \geq 0$'' are \emph{not} tight. For every $j\in\items$, completeness guarantees that $x_{i,j}>0$ for some $i\in[d]$. Thus, at least $m$ variables $x_{i,j}$ are strictly positive, and each chore shared between groups requires at least one additional strictly positive variable. Since at most $m+d(d-1)$ such variables are strictly positive, at most $d(d-1)$ chores are shared between groups of agents.

    Let $x^*$ be an optimal solution of LP~\eqref{LP_gap_chores} where at most $d(d-1)$ chores are shared between groups. 
    For each group $i$ and item type $z \in [t]$, let $B_{i,z} = n_i \sum_{j \in \items: \text{$j$ has type $z$}} x^*_{i,j}$ be the total mass of type-$z$
    chores that group $i$ received under $x^*$. Our aim is to make $B_{i,z}$s integral, without changing them too much. 

    For each $z\in[t]$, initialize $S_z \leftarrow 0$. We will have three phases of
    decreasing $B_{i,z}$s and increasing the corresponding $S_z$. After these phases, we will guarantee that for every $i \in [d]$ and $z \in [t]$, $B_{i,z} \equiv 0\pmod{n_i}$. Therefore, 
    $B_{i,z}$ chores of type $z$ can be integrally and equally allocated to the individual agents of each group $i$, for every $i \in [d]$.
    We will also guarantee that $S_z = \sum_{i \in [d]} y_{i,z} \cdot n_i$ for some set of non-negative integers $\{y_{i,z}\}_{i \in [d]}$; this implies that $S_z$ chores of type $z$
    can be integrally distributed among groups in a way that, for each group $i$, individuals of group $i$ receive the same number of chores.

    \emph{Phase 1.} 
    Consider every item $\ell \in \items$ such that there exist $i, i' \in [d]$ (with $i \neq i'$) with $x^*_{i,\ell}>0$ and $x^*_{i',\ell}>0$; that is, $\ell$ is an item that is fractionally shared among at least two groups of agents. For every such item $\ell$, set $B_{i,z} \leftarrow B_{i,z} - n_i x^*_{i,\ell}$, for all $i \in [d]$, and $S_z \leftarrow S_z + \sum_{i \in [d]}  n_i x^*_{i,\ell}$, where $z$ is the type of item $\ell$. From our discussion above, we have that there are at most $d(d-1)$ such items $\ell$. Therefore, at the end of Phase 1, we have that $\sum_{z \in [t]} S_z \le d(d-1)$.

    \emph{Phase 2.} At the start of Phase 2, we have that for every item $j \in \items$ of type $z$, $B_{i,z}$ is an integer, for all $i \in [d]$. Let $B_{i,z} = n_i \cdot q_{i,z} +r_{i,z}$, where $r_{i,z} \in [0, n_i-1]$. 
    For every $z \in [t]$, set $B_{i,z} \leftarrow B_{i,z} - r_{i,z}$ for every $i \in [d]$, and set $S_z \leftarrow S_z + \sum_{i \in [d]} r_{i,z}$. Therefore, at the end of Phase 2, we have that $B_{i,z} = n_i \cdot q_{i,z}$ for some integer $q_{i,z}$, and that $\sum_{z \in [t]} S_z \le d(d-1) +  t\sum_{i \in [d]}(n_i-1) = d(d-1) + t(n-d)$.

    \emph{Phase 3.}  For every $z \in [t]$ while $S_z < \theta$, we pick an arbitrary group of agents $i \in [d]$, such that $B_{i,z} > 0$, and 
    set $B_{i,z} \leftarrow B_{i,z} - n_i$ and $S_z \leftarrow S_z + n_i$. First, note that it cannot be that $B_{i,z} = 0$ for all $i \in [d]$ and $S_z < \theta$, since $k_z = \sum_{i\in [d]} B_{i,z} + S_z \geq \theta$; therefore, if $S_z < \theta$ then $B_{i,z} > 0$. Moreover, at the end of Phase 2, we are guaranteed that $B_{i,z} = n_i \cdot q_{i,z}$, for some non-negative integer $q_{i,z}$; therefore $B_{i,z} \leftarrow B_{i,z} - n_i$ is a valid operation. Repeating until $S_z \geq \theta$ for all $z \in [t]$, we have that $\sum_{z \in [t]} S_z \le d(d-1) + t(n-d)+t(\theta+n_d-1)$. We also have that for each $z \in [t]$, $B_{i,z} = n_i \cdot q'_{i,z}$ for some set of integers $\{q'_{i,z}\}_{i \in [d]}$. Thus, $\sum_{i \in [d]} B_{i,z} \equiv 0 \pmod{g}$. But since $k_z \equiv 0 \pmod{g}$ and $S_z + \sum_{i \in [d]} B_{i,z} = k_z$, we also have that $S_z \equiv 0 \pmod{g}$ for all $z \in [t]$.

    Since $S_z \equiv 0 \pmod{g}$ and $S_z \geq \theta$, applying $\Cref{corollary: frobenius coin problem}$ we have that there exists a set of non-negative integers $\{y_{i,z}\}_{i \in [d]}$ such that $S_z = \sum_{i \in [d]} y_{i,z} \cdot n_i$. We can therefore allocate $y_{i,z}$ chores of type $z$ to each agent in group $i$, for all $i \in [d]$, $z \in [t]$. 
    Since $B_{i,z} = n_i \cdot q'_{i,z}$ for some set of integers $\{q'_{i,z}\}_{i \in [d]}$, we can also allocate $q'_{i,z}$ chores of type $z$ to each agent in group $i$. Let $A$ be this allocation, and notice that $A$ allocates all chores.
    We have that $\sum_{z \in [t]} y_{i,z} \le \sum_{z \in [t]} S_z \le d(d-1) + t(n-d)+t(\theta+n_d-1)$. Thus, after Phase 3, every agent has received at most $d(d-1) + t(n-d)+t(\theta+n_d-1)$ extra chores (compared to the allocation $x^*$).
    All agents of the same group have received exactly the same integral bundle. For every $i\in[d]$, the total cost added to any agent's bundle, and the total cost removed from any agent's bundle, measured by $c_i$, are each at most $(d(d-1)+t(\theta+n+n_d-d-1)) \max_{j\in\items}c_{i,j}(\kb)$. Thus $\min_{ i' \in [d], i'\ne i} \cgap_{i,i'}(A) \ge  \min_{ i' \in [d], i'\ne i } \cgap_{i,i'}(x^*)  - 2(d(d-1) +t(\theta+n + n_d - d -1))\max_{j \in \items} c_{i,j}(\kb) \geq \norm{c_i(\kb)}_1 \cdot\alpha^* - 2(d(d-1) +t(\theta+n + n_d - d -1))\max_{j \in \items} c_{i,j}(\kb)$.
\end{proof}

\begin{proof}[Proof of~\Cref{lemma: copies lemma chores}]
    It is easy to verify that $\alpha \norm{c_i}_1 \le \norm{c_i(\kb)}_1 \le \beta \norm{c_i}_1$. Thus there exist $\alpha\le \gamma_{i'} \le\beta$ and $\alpha\le \gamma_{i} \le\beta$ such that $\norm{c_i(\kb)}_1 = \gamma_i \cdot \norm{c_i}_1$ and $\norm{c_{i'}(\kb)}_1 = \gamma_{i'} \cdot \norm{c_{i'}}_1$.
    \begin{align*}
        &D_{KL}(\normc_i(\kb) || \normc_{i'}(\kb)) = \sum_{j \in \items} \normc_{i,j}(\kb)\ln \left(\frac{\normc_{i,j}(\kb)}{\normc_{i',j}(\kb)} \right)\\
        &\quad = \sum_{j \in \items} \normc_{i,j}(\kb)\ln \left(\frac{c_{i,j}(\kb)/\norm{c_{i}(\kb)}_1}{c_{i',j}(\kb)/\norm{c_{i'}(\kb)}_1} \right)\\
        &\quad = \sum_{j \in \items} \normc_{i,j}(\kb)\ln \left(\frac{c_{i,j}(\kb)/\norm{c_{i}}_1}{c_{i',j}(\kb)/\norm{c_{i'}}_1} \cdot \frac{\gamma_{i'}}{\gamma_i} \right)\\
        &\quad = \sum_{j \in \items} \normc_{i,j}(\kb)\ln \left( \frac{\gamma_{i'}}{\gamma_i} \right) + \sum_{j \in \items} \normc_{i,j}(\kb)\ln \left(\frac{c_{i,j}(\kb)/\norm{c_{i}}_1}{c_{i',j}(\kb)/\norm{c_{i'}}_1} \right)\\
        &\quad = \ln \left( \frac{\gamma_{i'}}{\gamma_i} \right) + \sum_{j \in \items} \frac{c_{i,j}(\kb)}{\norm{c_i(\kb)}_1}\ln \left(\frac{c_{i,j}(\kb)/\norm{c_{i}}_1}{c_{i',j}(\kb)/\norm{c_{i'}}_1} \right)\\
        &\quad = \ln \left( \frac{\gamma_{i'}}{\gamma_i} \right) + \frac{1}{\gamma_i} \sum_{z \in [t]} k_z \cdot \frac{c_{i,z}}{\norm{c_i}_1}\ln \left(\frac{c_{i,z}/\norm{c_{i}}_1}{c_{i',z}/\norm{c_{i'}}_1} \right)\\
        &\quad = \ln \left( \frac{\gamma_{i'}}{\gamma_i} \right) + \frac{1}{\gamma_i} \sum_{z \in [t]} k_z \cdot \normc_{i,z}\ln \left(\frac{\normc_{i,z}}{\normc_{i',z}} \right)\\
        &\quad = \ln \left( \frac{\gamma_{i'}}{\gamma_i} \right)+ \frac{1}{\gamma_i} \left(
        \begin{aligned}
            &\sum_{z \in [t]} k_z \cdot \normc_{i,z}\ln \left(\normc_{i,z}\right)\\[-.25ex]
            &{}- \sum_{z \in [t]} k_z \cdot \normc_{i,z}\ln \left(\normc_{i',z}\right)
        \end{aligned}
        \right )\\
        &\quad \ge \ln \left( \frac{\gamma_{i'}}{\gamma_i} \right)+ \frac{1}{\gamma_i} \left(
        \begin{aligned}
            &\beta \sum_{z \in [t]} \normc_{i,z}\ln \left(\normc_{i,z}\right) - \alpha \sum_{z \in [t]} \normc_{i,z}\ln \left(\normc_{i',z}\right)
        \end{aligned}
        \right ) \tag{$0<\normc_{i,z} \le 1$}\\
        &\quad = \ln \left( \frac{\gamma_{i'}}{\gamma_i} \right) + \frac{\alpha}{\gamma_i} \sum_{z \in [t]} \normc_{i,z}\ln \left(\frac{\normc_{i,z}}{\normc_{i',z}} \right) + \frac{\beta-\alpha}{\gamma_i}\sum_{z \in [t]} \normc_{i,z}\ln \left(\normc_{i,z}\right)\\
        &\quad \ge \ln \left( \frac{\gamma_{i'}}{\gamma_i} \right) + \frac{\alpha}{\gamma_i} \sum_{z \in [t]} \normc_{i,z}\ln \left(\frac{\normc_{i,z}}{\normc_{i',z}} \right) + \frac{\beta-\alpha}{\gamma_i}\sum_{z \in [t]} \normc_{i,z}\ln \left(\normc_{min}\right)\\
        &\quad \ge \frac{\alpha}{\beta} D_{KL}(\normc_i || \normc_{i'}) - \left(\ln \left( \frac{\beta}{\alpha} \right) + \frac{\beta-\alpha}{\alpha}\ln \left(\normc_{min}\right)^{-1} \right).
    \end{align*}
\end{proof}

\begin{proof}[Proof of~\Cref{lemma:technical inequality chores}]
    \begin{align*}
        &x\left(\frac{x}{x+a}b- \ln \frac{x+a}{x} \right) -c\ln(x+a) = x\left(\frac{x}{x+a}b- \ln \left(1+a/x \right) \right) -c\ln(x+a)\\
        &\quad \ge x\left(\frac{x}{x+a}b- \frac{a}{x} \right) -c\ln(x+a) \tag{$\ln(1+u) \le u$}\\
        &\quad \ge x\left(\frac{x-a}{x}b- \frac{a}{x} \right) -c\ln(x+a) \tag{$\frac{x-a}{x} \le \frac{x}{x+a}$}\\
        &\quad = xb - c\ln(x+a) -a(b+1)
    \end{align*}
    Now for function $f(x) = xb - c\ln(x+a)$, notice that $f''(x) = \frac{c}{(a+x)^2} \ge 0$. Thus, $f(x)$ is convex, which implies that it always lies above its tangent lines. We pick the tangent for $x=\frac{2c}{b}-a$ since it simplifies our bounds. Combining the above we have that:
    \begin{align*}
        &x\left(\frac{x}{x+a}b- \ln \frac{x+a}{x} \right) -c\ln(x+a) \ge xb - c\ln(x+a) -a(b+1) \\
        &\quad = f(x) - a(b+1) \\
        &\quad \ge f\left( \frac{2c}{b}-a\right) + f'\left( \frac{2c}{b}-a\right)\cdot \left( x - \left(\frac{2c}{b}-a \right) \right) - a(b+1) \\
        &\quad = 2c -ab- c\ln\left(\frac{2c}{b}\right) + \frac{b}{2}x - c +\frac{ab}{2} -ab-a \\
        &\quad = \frac{b}{2}x - \left(a\left(\frac{3}{2}b+1\right) + c\left(\ln\left(\frac{2c}{b}\right) - 1 \right) \right).
    \end{align*}
    
\end{proof}

\begin{proof}[Proof of~\Cref{application: chores copies}]
    
    The case $d=1$ is immediate, so assume $d\ge2$ for the remainder of the proof.

    Let
    \[
    \mu = \left\lceil 2\frac{
    \begin{aligned}
    &n+ \left( \frac{5}{2}n+\lambda-1 \right)
    \ln \left(\normc_{min}\right)^{-1} + \lambda\left(
    \ln\left(
    \frac{2\lambda}{
    \displaystyle\min_{\substack{\hat{i}, \tilde{i} \in [d]\\ \hat{i}\ne\tilde{i}}}
    D_{KL}(\normc_{\hat{i}} || \normc_{\tilde{i}})}
    \right)-1 \right)
    \end{aligned}
    }{
    \displaystyle\min_{\substack{\hat{i}, \tilde{i} \in [d]\\ \hat{i}\ne\tilde{i}}}
    D_{KL}(\normc_{\hat{i}} || \normc_{\tilde{i}})
    }\right\rceil.
    \]
    First, notice that, without loss of generality, we can assume that $k_z \in [\mu, \mu+n-1]$. If that is not the case, and $k_z > \mu+n-1$, we can simply allocate one copy of item $z$ to each agent and continue this process until the condition is satisfied; since all agents receive the exact same amount of items, envy is not affected by this process. Also, since we allocate copies in multiples of $n$ and $n \equiv 0 \pmod{g}$, then the remaining number of copies is also divisible by $g$ (so, the condition $k_z \equiv 0 \pmod{g}$ also remains true).

    From \Cref{theorem: Conditions for envy freeness chores} we know that if $\normc_{max}(\kb) \leq  \frac{ \min_{\hat{i}, \tilde{i} \in [d],\ \hat{i}\ne\tilde{i}} D_{KL}(\normc_{\hat{i}}(\kb) || \normc_{\tilde{i}}(\kb))}{ \lambda \ln \normc_{min}(\kb)^{-1}}$, then an envy-free allocation exists. Notice that $\normc_{max}(\kb) \le \frac{\normc_{max}}{\mu}\le \frac{1}{\mu}$. Thus, we conclude that if 
    \[\mu \cdot \min_{\hat{i}, \tilde{i} \in [d],\ \hat{i}\ne\tilde{i}} D_{KL}(\normc_{\hat{i}}(\kb) || \normc_{\tilde{i}}(\kb)) \ge  \lambda \ln \normc_{min}(\kb)^{-1}\] 
    an envy-free allocation exists. We also have that for any distinct $\hat{i}, \tilde{i} \in [d]$:
    \begin{align*}
        &0 \le \mu \frac{ D_{KL}(\normc_{\hat{i}} || \normc_{\tilde{i}})}{2} - \left( n+ \left( \frac{5}{2}n+\lambda-1 \right)\ln \left(\normc_{min}\right)^{-1} + \lambda\left(\ln\left(\frac{2\lambda}{ D_{KL}(\normc_{\hat{i}} || \normc_{\tilde{i}})}\right)-1 \right)
        \right) \\
        &\quad = \mu \frac{ D_{KL}(\normc_{\hat{i}} || \normc_{\tilde{i}})}{2} - n\left( \frac{3}{2}\ln \left(\normc_{min}\right)^{-1} +1 \right) - (n-1)\ln \left(\normc_{min}\right)^{-1}\\
        &\qquad - \lambda \ln \left(\normc_{min}\right)^{-1} - \lambda\left(\ln\left(\frac{2\lambda}{ D_{KL}(\normc_{\hat{i}} || \normc_{\tilde{i}})}\right)-1 \right)\\
        &\quad \le \mu \frac{ D_{KL}(\normc_{\hat{i}} || \normc_{\tilde{i}})}{2} - (n-1)\left( \frac{3 D_{KL}(\normc_{\hat{i}} || \normc_{\tilde{i}})}{2} +1 \right) - (n-1)\ln \left(\normc_{min}\right)^{-1}\\
        &\qquad - \lambda \ln \left(\normc_{min}\right)^{-1} - \lambda\left(\ln\left(\frac{2\lambda}{ D_{KL}(\normc_{\hat{i}} || \normc_{\tilde{i}})}\right)-1 \right) \tag{$ D_{KL}(\normc_{\hat{i}} || \normc_{\tilde{i}}) \le \ln\left(\normc_{min}\right)^{-1}$}\\
        &\quad \le \mu \left(\frac{\mu}{\mu+n-1} D_{KL}(\normc_{\hat{i}} || \normc_{\tilde{i}}) - \ln \left(\frac{\mu +n -1}{\mu} \right) \right)\\
        &\qquad - \lambda \ln(\mu+n-1) -\lambda \ln \left(\normc_{min}\right)^{-1} - (n-1)\ln \left(\normc_{min}\right)^{-1} \tag{\Cref{lemma:technical inequality chores}}\\
        &\quad = \mu \left( \frac{\mu}{\mu+n-1} D_{KL}(\normc_{\hat{i}} || \normc_{\tilde{i}}) - \ln \left(\frac{\mu +n -1}{\mu} \right) - \frac{n-1}{\mu}\ln \left(\normc_{min}\right)^{-1}
        \right) - \lambda \ln \left(\frac{\normc_{min}}{\mu +n -1}\right)^{-1} \\
        &\quad \le \mu D_{KL}(\normc_{\hat{i}}(\kb) || \normc_{\tilde{i}}(\kb)) - \lambda \ln \left(\frac{\normc_{min}}{\mu +n -1}\right)^{-1} \tag{\Cref{lemma: copies lemma chores}}\\
        &\quad \le \mu D_{KL}(\normc_{\hat{i}}(\kb) || \normc_{\tilde{i}}(\kb)) - \lambda \ln \left(\normc_{min}(\kb)\right)^{-1} \tag{$\normc_{min}(\kb)\ge \frac{\normc_{min}}{\mu +n -1}$}
    \end{align*}
    Combining the above, we get that an envy-free allocation exists.
\end{proof}

\section{Proofs missing from~\Cref{section: cake cutting}}\label{app: missing from cake cutting}

\begin{proof}[Proof of~\Cref{lemma: Discrete cake eucledian bound}]

    For $x\in[0,1)$, define
\[
\bar f_i(x)
=\frac{1}{\epsilon}
\int_{\floor{x/\epsilon}\epsilon}^{
(\floor{x/\epsilon}+1)\epsilon}
\normtwou_i(y)\,dy,
\]
where $\bar f_i(1)
:=\frac{1}{\epsilon}\int_{1-\epsilon}^{1}
\normtwou_i(y)\,dy$. For $x\in[0,1)$, we have:

    \begin{align*}
        &|\normtwou_i(x) - \bar{f}_i(x)| = \left|
            \normtwou_i(x) - \frac{1}{\epsilon} \int_{\floor{x/\epsilon}\epsilon}^{\left(\floor{x/\epsilon}+1\right)\epsilon} \normtwou_i(y) dy
        \right| \\
        &\quad= \left| \frac{1}{\epsilon} \int_{\floor{x/\epsilon}\epsilon}^{\left(\floor{x/\epsilon}+1\right)\epsilon}\normtwou_i(x) dy 
            - \frac{1}{\epsilon} \int_{\floor{x/\epsilon}\epsilon}^{\left(\floor{x/\epsilon}+1\right)\epsilon} \normtwou_i(y) dy
        \right| \\
        &\quad=\left| \frac{1}{\epsilon} \int_{\floor{x/\epsilon}\epsilon}^{\left(\floor{x/\epsilon}+1\right)\epsilon} \normtwou_i(x) -\normtwou_i(y)dy \right|\\
        &\quad\le \frac{1}{\epsilon} \int_{\floor{x/\epsilon}\epsilon}^{\left(\floor{x/\epsilon}+1\right)\epsilon} |\normtwou_i(x) -\normtwou_i(y)|dy\\
        &\quad\le \frac{1}{\epsilon} \int_{\floor{x/\epsilon}\epsilon}^{\left(\floor{x/\epsilon}+1\right)\epsilon} k|x -y|dy\\
        &\quad= \frac{k}{\epsilon}\left( \int_{\floor{x/\epsilon}\epsilon}^{x} (x-y)dy + \int_{x}^{\left(\floor{x/\epsilon}+1\right)\epsilon} (y-x) dy
        \right)\\
        &\quad= \frac{k}{\epsilon}\left( x\left(x - \floor{x/\epsilon}\epsilon \right) - x\left( \left(\floor{x/\epsilon} +1\right)\epsilon - x\right) -\frac{1}{2} \left[x^2 - \left( \floor{x/\epsilon}\epsilon \right)^2\right]\ +\frac{1}{2} \left[ \left( \left(\floor{x/\epsilon} +1 \right)\epsilon \right)^2 - x^2\right]
        \right) \\
        &\quad =\frac{k}{\epsilon}\left(
        x\left(2x - \left(2\floor{x/\epsilon} +1\right)\epsilon \right) - x^2 + \frac{1}{2} \left(\left( \floor{x/\epsilon}\epsilon \right)^2 + \left( \left(\floor{x/\epsilon} +1 \right)\epsilon \right)^2 \right)
        \right)\\
        &\quad = \frac{k}{\epsilon}\left(
            x^2 - x\left(2\floor{x/\epsilon} +1\right)\epsilon + \left( \floor{x/\epsilon}\epsilon \right)^2 + \floor{x/\epsilon}\epsilon^2 + \frac{1}{2} \epsilon^2
        \right)\\
        &\quad = \frac{k}{\epsilon}\left(\left( x- \floor{x/\epsilon} \epsilon - \frac{\epsilon}{2} \right)^2 + \frac{\epsilon^2}{4}\right)\\
        &\quad \le \frac{k \epsilon}{2}.
    \end{align*}

    \[\sqrt{\int_0^1 \left(\normtwou_i(x) -\bar{f}_i(x)\right)^2dx} \le  \sqrt{\int_0^1 \left(\frac{k \epsilon}{2}\right)^2dx} = \left(\frac{k \epsilon}{2}\right)\]

    We have the following:

    \begin{align*}
        \sqrt{\int_0^1 (\normtwou_i(x))^2dx} &= \sqrt{\int_0^1 (\bar{f}_i(x) + (\normtwou_i(x)-\bar{f}_i(x)))^2dx} \\
        & \le \sqrt{\int_0^1 (\bar{f}_i(x))^2dx} + \sqrt{\int_0^1 (\normtwou_i(x)-\bar{f}_i(x))^2dx} \tag{Minkowski's Inequality}\\
        & \le \sqrt{\int_0^1 (\bar{f}_i(x))^2dx} + \sqrt{\int_0^1 \left(\frac{k \epsilon}{2}\right)^2dx}\\
        & = \sqrt{\int_0^1 (\bar{f}_i(x))^2dx} + \frac{k \epsilon}{2}
    \end{align*}
    Thus $\sqrt{\int_0^1 (\bar{f}_i(x))^2dx} \ge  \sqrt{\int_0^1 (\normtwou_i(x))^2dx} - \frac{k \epsilon}{2} = 1-\frac{k \epsilon}{2}$. Similarly we can get that:
    \begin{align*}
        \sqrt{\int_0^1 (\bar{f}_i(x))^2dx} &= \sqrt{\int_0^1 (\normtwou_i(x) + (\bar{f}_i(x)-\normtwou_i(x)))^2dx} \\
        &\le \sqrt{\int_0^1 (\normtwou_i(x))^2dx} + \sqrt{\int_0^1 (\bar{f}_i(x)-\normtwou_i(x))^2dx} \tag{Minkowski's Inequality}\\
        &\le \sqrt{\int_0^1 (\normtwou_i(x))^2dx} + \sqrt{\int_0^1 \left(\frac{k \epsilon}{2}\right)^2dx}\\
        &= 1 + \frac{k \epsilon}{2}
    \end{align*}

    \begin{align*}
        &\int_{0}^{1}(\normtwou_i(x)-\normtwou_{i'}(x))^2dx\\
        &\qquad - \sum_{j=0}^{1/\epsilon-1} \left(
        \begin{aligned}
            &\frac{\int_{\epsilon j}^{\epsilon(1+j)}\normtwou_{i}(x)dx}{\sqrt{\sum_{j'=0}^{1/\epsilon-1} \left(\int_{\epsilon j'}^{ \epsilon ( j'+1)}\normtwou_i(x)dx\right)^2}} - \frac{\int_{\epsilon j}^{\epsilon(1+j)}\normtwou_{i'}(x)dx}{\sqrt{\sum_{j'=0}^{1/\epsilon-1} \left(\int_{\epsilon j'}^{ \epsilon ( j'+1)}\normtwou_{i'}(x)dx\right)^2}}
        \end{aligned}
        \right)^2 \\
        &\quad = \int_{0}^{1}(\normtwou_i(x))^2+ (\normtwou_{i'}(x))^2 -2\normtwou_i(x) \cdot \normtwou_{i'}(x)dx \\
        &\qquad - \sum_{j=0}^{1/\epsilon-1} \left( \frac{\int_{\epsilon j}^{\epsilon(1+j)}\normtwou_{i}(x)dx}{\sqrt{\sum_{j'=0}^{1/\epsilon-1} \left(\int_{\epsilon j'}^{ \epsilon ( j'+1)}\normtwou_i(x)dx\right)^2}} - \frac{\int_{\epsilon j}^{\epsilon(1+j)}\normtwou_{i'}(x)dx}{\sqrt{\sum_{j'=0}^{1/\epsilon-1} \left(\int_{\epsilon j'}^{ \epsilon ( j'+1)}\normtwou_{i'}(x)dx\right)^2}}
        \right)^2\\
        &\quad = 2- 2 \int_{0}^{1}\normtwou_i(x) \cdot \normtwou_{i'}(x)dx\\
        &\qquad - \sum_{j=0}^{1/\epsilon-1} \left( \frac{\int_{\epsilon j}^{\epsilon(1+j)}\normtwou_{i}(x)dx}{\sqrt{\sum_{j'=0}^{1/\epsilon-1} \left(\int_{\epsilon j'}^{ \epsilon ( j'+1)}\normtwou_i(x)dx\right)^2}} - \frac{\int_{\epsilon j}^{\epsilon(1+ j)}\normtwou_{i'}(x)dx}{\sqrt{\sum_{j'=0}^{1/\epsilon-1} \left(\int_{\epsilon j'}^{ \epsilon ( j'+1)}\normtwou_{i'}(x)dx\right)^2}}
        \right)^2 \\
        &\quad = 2- 2 \int_{0}^{1}\normtwou_i(x) \cdot \normtwou_{i'}(x)dx\\
        &\qquad - \sum_{j=0}^{1/\epsilon -1}\left( \frac{\left(\int_{\epsilon j}^{\epsilon(1+j)}\normtwou_{i}(x)dx\right)^2}{\sum_{j'=0}^{1/\epsilon-1} \left(\int_{\epsilon j'}^{ \epsilon ( j'+1)}\normtwou_i(x)dx\right)^2} + \frac{\left(\int_{\epsilon j}^{\epsilon(1+j)}\normtwou_{i'}(x)dx\right)^2}{\sum_{j'=0}^{1/\epsilon-1} \left(\int_{\epsilon j'}^{ \epsilon ( j'+1)}\normtwou_{i'}(x)dx\right)^2} \right. \\
        &\qquad \left. -2 \frac{\int_{\epsilon j}^{\epsilon(1+j)}\normtwou_{i}(x)dx \cdot \int_{\epsilon j}^{\epsilon(1+j)}\normtwou_{i'}(x)dx}{\sqrt{ \sum_{j'=0}^{1/\epsilon-1} \left(\int_{\epsilon j'}^{ \epsilon ( j'+1)}\normtwou_{i}(x)dx\right)^2 \cdot \sum_{j'=0}^{1/\epsilon-1} \left(\int_{\epsilon j'}^{ \epsilon ( j'+1)}\normtwou_{i'}(x)dx\right)^2}}
        \right) \\
        &\quad = 2\left( \sum_{j=0}^{1/\epsilon-1}\frac{\int_{\epsilon j}^{\epsilon(1+j)}\normtwou_{i}(x)dx \cdot \int_{\epsilon j}^{\epsilon(1+j)}\normtwou_{i'}(x)dx}{\sqrt{ \sum_{j'=0}^{1/\epsilon-1} \left(\int_{\epsilon j'}^{ \epsilon ( j'+1)}\normtwou_{i}(x)dx\right)^2 \cdot \sum_{j'=0}^{1/\epsilon-1} \left(\int_{\epsilon j'}^{ \epsilon ( j'+1)}\normtwou_{i'}(x)dx\right)^2}} - \int_{0}^{1}\normtwou_i(x) \cdot \normtwou_{i'}(x)dx
        \right) \\
        &\quad = 2\left( \sum_{j=0}^{1/\epsilon-1}\frac{\epsilon^2 \bar{f}_i(\epsilon j) \cdot \bar{f}_{i'}(\epsilon j)}{\sqrt{ \sum_{j'=0}^{1/\epsilon-1} \left(\int_{\epsilon j'}^{ \epsilon ( j'+1)}\normtwou_{i}(x)dx\right)^2 \cdot \sum_{j'=0}^{1/\epsilon-1} \left(\int_{\epsilon j'}^{ \epsilon ( j'+1)}\normtwou_{i'}(x)dx\right)^2}} - \int_{0}^{1}\normtwou_i(x) \cdot \normtwou_{i'}(x)dx
        \right) \\
        &\quad =2\left(
        \begin{aligned}
            &\frac{\epsilon \sum_{j=0}^{1/\epsilon-1} \int_{\epsilon j}^{\epsilon(1+j)} \bar{f}_i(\epsilon j) \cdot \bar{f}_{i'}(\epsilon j)dx}{\sqrt{ \sum_{j'=0}^{1/\epsilon-1} \left(\int_{\epsilon j'}^{ \epsilon ( j'+1)}\normtwou_{i}(x)dx\right)^2 \cdot \sum_{j'=0}^{1/\epsilon-1} \left(\int_{\epsilon j'}^{ \epsilon ( j'+1)}\normtwou_{i'}(x)dx\right)^2}} - \int_{0}^{1}\normtwou_i(x) \cdot \normtwou_{i'}(x)dx
        \end{aligned}
        \right)\\
        &\quad =2\left(\frac{\epsilon \sum_{j=0}^{1/\epsilon-1} \int_{\epsilon j}^{\epsilon(1+j)} \bar{f}_i(x) \cdot \bar{f}_{i'}(x)dx}{\sqrt{
                \sum_{j'=0}^{1/\epsilon-1} \left(\int_{\epsilon j'}^{ \epsilon ( j'+1)}\normtwou_{i}(x)dx\right)^2 \cdot \sum_{j'=0}^{1/\epsilon-1} \left(\int_{\epsilon j'}^{ \epsilon ( j'+1)}\normtwou_{i'}(x)dx\right)^2}} - \int_{0}^{1}\normtwou_i(x) \cdot \normtwou_{i'}(x)dx
        \right)\\
        &\tag{$\bar f_i(\epsilon j)=\bar f_i(x)$ for $x\in[\epsilon j,\epsilon(j+1))$}\\
        &\quad = 2\left(
        \frac{\epsilon \int_{0}^{1} \bar{f}_i(x) \cdot \bar{f}_{i'}(x)dx}{\sqrt{\sum_{j'=0}^{1/\epsilon-1} \left(\int_{\epsilon j'}^{ \epsilon ( j'+1)}\normtwou_{i}(x)dx\right)^2 \cdot \sum_{j'=0}^{1/\epsilon-1} \left(\int_{\epsilon j'}^{ \epsilon ( j'+1)}\normtwou_{i'}(x)dx\right)^2}} - \int_{0}^{1}\normtwou_i(x) \cdot \normtwou_{i'}(x)dx
        \right)\\
        &\quad = 2\left(
        \frac{\epsilon \int_{0}^{1} \bar{f}_i(x) \cdot \bar{f}_{i'}(x)dx}{\sqrt{
                \sum_{j'=0}^{1/\epsilon-1} \left(\epsilon \bar{f}_i(\epsilon j') \right)^2 \cdot \sum_{j'=0}^{1/\epsilon-1} \left(\epsilon \bar{f}_{i'}(\epsilon j') \right)^2 }} - \int_{0}^{1}\normtwou_i(x) \cdot \normtwou_{i'}(x)dx
        \right) \\
        &\quad =2\left(
        \frac{\epsilon \int_{0}^{1} \bar{f}_i(x) \cdot \bar{f}_{i'}(x)dx}{\sqrt{    
                \epsilon\sum_{j'=0}^{1/\epsilon-1} \int_{\epsilon j'}^{\epsilon(j'+1)}\left(\bar{f}_i(\epsilon j') \right)^2dx \cdot \epsilon \sum_{j'=0}^{1/\epsilon-1} \int_{\epsilon j'}^{\epsilon(1+j')}\left( \bar{f}_{i'}(\epsilon j') \right)^2 \,dx
            }} - \int_{0}^{1}\normtwou_i(x) \cdot \normtwou_{i'}(x)dx
        \right)\\
        &\quad = 2\left(
            \frac{\int_{0}^{1} \bar{f}_i(x) \cdot \bar{f}_{i'}(x)dx}{\sqrt{ \sum_{j'=0}^{1/\epsilon-1} \int_{\epsilon j'}^{\epsilon(j'+1)}\left(\bar{f}_i(x) \right)^2dx \cdot \sum_{j'=0}^{1/\epsilon-1} \int_{\epsilon j'}^{\epsilon(1+j')}\left( \bar{f}_{i'}(x) \right)^2 \,dx
            }} - \int_{0}^{1}\normtwou_i(x) \cdot \normtwou_{i'}(x)dx
        \right)\\
        &\quad = 2\left( \frac{\int_{0}^{1} \bar{f}_i(x) \cdot \bar{f}_{i'}(x)dx}{\sqrt{ \int_{0}^{1}\left(\bar{f}_i(x) \right)^2dx \cdot \int_{0}^{1}\left( \bar{f}_{i'}(x) \right)^2 \,dx
            }}  - \int_{0}^{1}\normtwou_i(x) \cdot \normtwou_{i'}(x)dx
        \right).
    \end{align*}

The function $\bar{f}_i$ is the $\ell_2$-orthogonal projection of $\normtwou_i$ onto the functions that are constant on each piece. Therefore,
\[
\norm{\bar{f}_i}_2^2
=1-\norm{\normtwou_i-\bar{f}_i}_2^2
\]
and
\[
\left\langle\normtwou_i,\normtwou_{i'}\right\rangle
=\left\langle\bar{f}_i,\bar{f}_{i'}\right\rangle
+\left\langle\normtwou_i-\bar{f}_i,
\normtwou_{i'}-\bar{f}_{i'}\right\rangle.
\]
Since $0\le\langle\bar{f}_i,\bar{f}_{i'}\rangle\le
\norm{\bar{f}_i}_2\norm{\bar{f}_{i'}}_2$ and $1-\norm{\bar{f}_i}_2\le
\norm{\normtwou_i-\bar{f}_i}_2^2$, the difference in the preceding calculation is at most
\begin{align*}
&2\left(1-\norm{\bar{f}_i}_2\norm{\bar{f}_{i'}}_2\right) +2\norm{\normtwou_i-\bar{f}_i}_2
\norm{\normtwou_{i'}-\bar{f}_{i'}}_2 \\
&\qquad \le 2\norm{\normtwou_i-\bar{f}_i}_2^2
+2\norm{\normtwou_{i'}-\bar{f}_{i'}}_2^2 +2\norm{\normtwou_i-\bar{f}_i}_2
\norm{\normtwou_{i'}-\bar{f}_{i'}}_2\\
&\qquad\le 6\left(\frac{k\epsilon}{2}\right)^2\\
&\qquad\le \frac{7k\epsilon}{2},
\end{align*}
where the last inequality uses $k\epsilon\le1$. This proves the lemma.
 \end{proof}

\section{Proofs missing from~\Cref{sec: proportionality}}\label{app: missing from proportionality}

\begin{proof}[Proof of~\Cref{lemma: stochastic prop}]
     From the variational characterization of chi-squared divergence, we know that $D_{\chi^2}(\normv_i || S) \ge \langle \normv_i, g \rangle - \frac{\langle S, g^2 \rangle}{4}-1$, where $g$ is any vector in $\mathbb{R}^m$ and $g^2 = [g_1^2, g_2^2, \cdots, g_m^2]$. By choosing $g = 2m\cdot \normv_i$, we have that $D_{\chi^2}(\normv_i || S) \ge \langle \normv_i, 2m\cdot \normv_i \rangle - \frac{\langle S, (2m\cdot \normv_i)^2 \rangle}{4}-1$.

     For compactness, let
     \[
     \Delta:=\sqrt{2m\log n}
     \qquad\text{and}\qquad
     q:=\frac{1}{n^4}.
     \]
     From Hoeffding's bound, and using the fact that $\mathbb{E}[v^2_{i,j}] = \frac{1}{3}$ and $\mathbb{E}[v^3_{i,j}] = \frac{1}{4}$, we have the following bounds, for all $i \in \agents$: (i) $\PP\left[\sum_{j' \in \items}v_{i,j'} \ge \frac{m}{2} + \Delta\right] \le q$, (ii) $\PP\left[\sum_{j' \in \items}v^2_{i,j'} \le \frac{m}{3} - \Delta\right] \le q$, (iii) $\PP\left[\sum_{j' \in \items}v^2_{i,j'} \ge \frac{m}{3} + \Delta\right] \le q$, as well as (iv) $\PP\left[\sum_{j' \in \items}v^3_{i,j'} \ge \frac{m}{4} + \Delta\right] \le q$.

     Also from Hoeffding's bound we have that for any $i' \in \agents \setminus \{i\}$:$\PP\left[\sum_{j' \in \items}v_{i',j'} \le m/2 - \Delta\right] \le q$ and $\PP\left[\sum_{j' \in \items} v_{i',j'} \cdot v^2_{i,j'} \ge m/6 + \Delta\right] \le q$.
    
     There are $O(n^2)$ such events, so by a union bound none of them occurs with probability $1-O(1/n^2)$; conditioning on this, we have:
     \begin{align*}
         &D_{\chi^2}(\normv_i || S) \ge \langle \normv_i, 2m\cdot \normv_i \rangle
         - \frac{\langle S, (2m\cdot \normv_i)^2 \rangle}{4}-1 \\
         &= \sum_{j \in \items} 2m \cdot (\normv_{i,j})^2  - \sum_{j \in \items}m^2 \cdot
         S_j(\normv_{i,j})^2 -1\\
         &= 2m \sum_{j \in \items} \frac{v^2_{i,j}}{\left(\sum_{j' \in \items}v_{i,j'}\right)^2 }  - \frac{m^2}{n} \sum_{j \in \items} \sum_{i' \in \agents} \normv_{i',j} \cdot (\normv_{i,j})^2 - 1\\
         &=2m \sum_{j \in \items} \frac{v^2_{i,j}}{\left(\sum_{j' \in \items}v_{i,j'}\right)^2}  - \frac{m^2}{n} \sum_{j \in \items} \sum_{i' \in \agents\setminus\{i\}}
         \frac{v_{i',j}}{\sum_{j' \in \items}v_{i',j'}}
         \cdot \frac{v^2_{i,j}}{\left(\sum_{j' \in \items}v_{i,j'}\right)^2}  \\
         &\qquad - \frac{m^2}{n} \sum_{j\in \items}\frac{v^3_{i,j}}{\left(\sum_{j' \in \items}v_{i,j'}\right)^3} - 1\\
         &\ge \frac{2m}{\left(m/2 + \Delta\right)^2}\sum_{j \in \items} v^2_{i,j}  -\frac{m^2}{n\left(m/2 - \Delta\right)^3}
         \sum_{j \in \items}\sum_{i' \in \agents\setminus\{i\}} v_{i',j} v^2_{i,j}  - \frac{m^2}{n\left(m/2 - \Delta\right)^3}\sum_{j \in \items} v^3_{i,j} -1\\
         &\ge \frac{2m \left(m/3 - \Delta \right)}{\left(m/2 + \Delta\right)^2}  - \frac{m^2 (n-1)\left(m/6 + \Delta \right)}{n\left(m/2 - \Delta\right)^3} - \frac{m^2\left(m/4 + \Delta\right)}{n\left(m/2 - \Delta\right)^3} - 1\\
         &= \frac{8}{3}\frac{1 - 3\Delta/m}{\left(1 + 2\Delta/m\right)^2}  - \frac{4(n-1)}{3n}\frac{1 + 6\Delta/m}{\left(1 - 2\Delta/m\right)^3}  - \frac{2}{n}\frac{1 + 4\Delta/m}{\left(1 - 2\Delta/m\right)^3} - 1.
     \end{align*}

     Let $\delta = \Delta/m$.
     Then, we have that 
     \begin{align*}
         D_{\chi^2}(\normv_i || S) &\ge \frac{8}{3}\frac{1 - 3\delta}{\left(1 + 2\delta\right)^2}  - \frac{4(n-1)}{3n}\frac{1 + 6\delta }{\left(1 - 2\delta\right)^3}  - \frac{2}{n}\frac{1 + 4\delta}{\left(1 - 2\delta\right)^3} - 1 \\
         &\geq \frac{8}{3}\frac{1 - 3\delta}{\left(1 + 2\delta\right)^2}  - \frac{1 + 6\delta}{\left(1 - 2\delta\right)^3}
         \left( \frac{4(n-1)}{3n} + \frac{2}{n} \right) - 1 \\
         &= \frac{8}{3}\frac{1 - 3\delta}{\left(1 + 2\delta\right)^2} - \frac{1 + 6\delta}{\left(1 - 2\delta\right)^3}
         \left( \frac{4n+2}{3n} \right) - 1.
     \end{align*}

    For $m \geq 2(11 z - 11)^2 \log(n)$ we have that $\frac{1 - 3\delta}{\left(1 + 2\delta\right)^2} \geq \frac{z-1}{z}$ ($z \geq 2$), and for $m \geq 2(12 z + 6)^2 \log(n)$ we have that $\frac{1 + 6\delta}{\left(1 - 2\delta\right)^3} \leq \frac{z+1}{z}$. Therefore, for $m \geq 2\max\{ (11 z - 11)^2, (12 z + 6)^2 \} \log(n)$, we have that
    \begin{align*}
        D_{\chi^2}(\normv_i || S)
        &\ge \frac{8}{3} \frac{z-1}{z} - \frac{z+1}{z} \left( \frac{4n+2}{3n} \right) - 1 \\
        &= \frac{8(z-1)n - 4n(z+1) - 2(z+1) - 3nz}{3nz} \\
        &= \frac{n (z-12) - 2(z+1)}{3nz}.
    \end{align*}
     By picking $z = 13$, we have that, for
     $m \geq 52488\log(n)$,
     $D_{\chi^2}(\normv_i || S) \geq \frac{n-28}{39n}$.
\end{proof}

\begin{proof}[Proof of~\Cref{lemma: disutility under TP}]
\begin{align*}
    \frac{c_i(x_i)}{c_i(\items)}
    &= \frac{1}{c_i(\items)} \sum_{j \in \items } c_{i,j}
    \frac{1/\normc_{i,j}}{\sum_{i' \in \agents} 1/\normc_{i',j}} \\
    &= \sum_{j \in \items }\frac{1}{\sum_{i' \in \agents} 1/\normc_{i',j}} \\
    &= \frac{1}{n}\sum_{j \in \items}\frac{n}{\sum_{i' \in \agents}1/\normc_{i', j}} \\
    &=  \frac{1}{n}\sum_{j \in \items} H_j  \tag{$H_j = \frac{n}{\sum_{i' \in \agents}\frac{1}{\normc_{i', j}}}$}\\
    &= \frac{1}{n}\left( 1-1+ \sum_{j \in \items}(2H_j - H_j) 
    \right) \\
    &= \frac{1}{n}\left(
    1-\frac{1}{n} \sum_{\ell \in \agents} \sum_{j \in \items} \normc_{\ell, j} + \frac{1}{n} \sum_{\ell \in \agents} \sum_{j \in \items}2H_j
    -  \frac{1}{n}\sum _{j \in \items} H^2_j \frac{n}{H_j}
    \right) \tag{$\sum_{j \in \items} \normc_{i,j} = 1$}\\
    &= \frac{1}{n}\left(
    1-\frac{1}{n} \sum_{\ell \in \agents} \sum_{j \in \items} \normc_{\ell, j} + \frac{1}{n} \sum_{\ell \in \agents} \sum_{j \in \items}2H_j
    -  \frac{1}{n}\sum _{j \in \items} H^2_j
    \sum_{\ell \in \agents} \frac{1}{\normc_{\ell, j}}
    \right) \tag{$H_j = \frac{n}{\sum_{\ell \in \agents}\frac{1}{\normc_{\ell, j}}}$}\\
    &= \frac{1}{n}\left(1-\frac{1}{n} \sum_{\ell \in \agents} \sum_{j \in \items} \frac{(\normc_{\ell, j})^2 - 2H_j\normc_{\ell, j} + H_j^2}{\normc_{\ell, j}}
    \right) \\
    &= \frac{1}{n}\left( 1-\frac{1}{n} \sum_{\ell \in \agents} \sum_{j \in \items} \frac{ \left( \normc_{\ell, j} - H_j \right)^2}{\normc_{\ell, j}}
    \right).
\end{align*}

\end{proof}

\begin{proof}[Proof of~\Cref{lemma: chores variational charachterization}]
    \begin{align*}
        \sum_{j \in \items} \frac{(\normc_{i, j} - H_j)^2}{\normc_{i, j}} &= \sum_{j \in \items} \left( \frac{H_j^2}{\normc_{i,j}} - 2H_j+ \normc_{i,j} \right) \\
        &= \sum_{j \in \items} \normc_{i,j} \left( \left(\frac{H_j}{\normc_{i,j}} \right)^2 -2\frac{H_j}{\normc_{i,j}} +1 \right)
    \end{align*}

    Now consider the function $f(t) = t^2-2t+1$ and its Legendre-Fenchel transformation $f^*(s)$:
    \begin{align*}
        f^*(s) &= \sup_{t \in \mathbb{R}} st - f(t) = \sup_{t \in \mathbb{R}} st - t^2+2t-1 = \frac{s^2}{4}+s,
    \end{align*}
    where the last equality is derived by taking the first and second derivatives and observing that $st - t^2+2t-1$ is maximized for $t = s/2+1$. Thus $f(t) \ge st - f^*(s) = s(t-1) - \frac{s^2}{4}$. Setting $t = \frac{H_j}{\normc_{i,j}}$ and $s = g_j$ we have
    \begin{align*}
        f\left( \frac{H_j}{\normc_{i,j}}\right) &= \left(\frac{H_j}{\normc_{i,j}} \right)^2 -2\frac{H_j}{\normc_{i,j}} +1 \ge g_j\left(\frac{H_j}{\normc_{i,j}} -1\right) - \frac{1}{4} g^2_j.
    \end{align*}
    Combining all of the above we have that:
    \begin{align*}
        \sum_{j \in \items} \frac{(\normc_{i, j} - H_j)^2}{\normc_{i, j}} &= \sum_{j \in \items} \normc_{i,j} \left( \left(\frac{H_j}{\normc_{i,j}} \right)^2 - 2\frac{H_j}{\normc_{i,j}} +1 \right) \\
        &= \sum_{j \in \items} \normc_{i,j}f\left(\frac{H_j}{\normc_{i,j}} \right) \\
        &\ge \sum_{j \in \items} \normc_{i,j} \left( g_j\left(\frac{H_j}{\normc_{i,j}}-1\right) - \frac{1}{4} g^2_j \right) \\
        &=  \sum_{j \in \items}  g_j(H_{j}-\normc_{i,j}) - \frac{1}{4} \sum_{j \in \items} \normc_{i,j} g_j^2 \\
        &= \langle H-\normc_i, g \rangle - \frac{\langle \normc_i, g^2 \rangle}{4}.
    \end{align*}
    Interestingly, this lower bound on $\sum_{j \in \items} \frac{(\normc_{i, j} - H_j)^2}{\normc_{i, j}}$ is tight; this can be verified by choosing $g_j = 2\frac{H_j}{\normc_{i,j}}-2$.
\end{proof}

\begin{proof}[Proof of~\Cref{lemma: stochastic chores}]
Using standard Chernoff-Hoeffding bounds we have $$\PP\left[\sum_{j \in \items}c_{i,j} \le m/2 - \sqrt{m \log n}\right] \le 1/n^2,$$ for all $i\in \agents$. Therefore, with high probability, $\sum_{j \in \items \setminus \{j'\}} c_{i,j} \ge m/2 -1 - \sqrt{m \log n}$ for all $j' \in \items$. Now consider the random variable $G_j = \left( \prod_{i \in [n]} c_{i,j}\right)^{1/n}$. It holds that $0 \le G_j \le 1$, and $\EE[G_j] = \left(\frac{n}{n+1}\right)^{n}$. Therefore, by standard Chernoff-Hoeffding bounds, we have that $$\PP\left[\sum_{j \in \items}G_j \ge m\left( \frac{n}{n+1}\right)^n + \sqrt{m \log n}\right] \le 1/n^2.$$
Union bounding, with probability $1-O(1/n)$, all of the above inequalities hold.

\Cref{lemma: chores variational charachterization} gives us that $\sum_{j \in \items} \frac{(\normc_{i, j} - H_j)^2}{\normc_{i, j}} \ge \langle H-\normc_i, g \rangle - \frac{\langle \normc_i, g^2 \rangle}{4}$. Choosing $g_{j} = -\frac{1}{2}$ we have:

\begin{align*}
    &\sum_{j \in \items} \frac{(\normc_{i, j} - H_j)^2}{\normc_{i, j}} \ge \langle H-\normc_i, g \rangle - \frac{\langle \normc_i, g^2 \rangle}{4} \\
    &\qquad = \sum_{j \in \items} \frac{1}{2}(\normc_{i,j} - H_j)  - \frac{1}{4} \sum_{j \in \items} \frac{1}{4}\normc_{i,j} \\
    &\qquad= \frac{7}{16} \sum_{j \in \items} \normc_{i,j} - \frac{1}{2} \sum_{j \in \items} H_{j} \\
    &\qquad= \frac{7}{16} - \frac{1}{2} \sum_{j \in \items} H_{j} \\
    &\qquad\ge \frac{7}{16} - \frac{1}{2} \sum_{j \in \items}
    \left(\prod_{i \in [n]} \normc_{i,j} \right)^{1/n}  \tag{HM-GM inequality}\\
    &\qquad= \frac{7}{16} - \frac{1}{2} \sum_{j \in \items}
    \left(\prod_{i \in [n]} \frac{c_{i,j}}{\sum_{j' \in \items} c_{i,j'}} \right)^{1/n} \\
    &\qquad\ge \frac{7}{16}  - \frac{1}{2} \sum_{j \in \items}
    \left(\prod_{i \in [n]} \frac{c_{i,j}}{\sum_{j' \in \items \setminus \{j\}} c_{i,j'}} \right)^{1/n}  \tag{$0 \le c_{i,j} \le 1 $}\\
    &\qquad\ge \frac{7}{16}  - \frac{1}{2} \sum_{j \in \items}
    \left(\prod_{i \in [n]} \frac{c_{i,j}}{\frac{m}{2} - 1- \sqrt{m \log n}} \right)^{1/n}  \tag{Hoeffding high-probability bound}\\
    &\qquad= \frac{7}{16}  - \frac{1}{m-2 - 2\sqrt{m \log n}}
    \sum_{j \in \items} \left(\prod_{i \in [n]} c_{i,j} \right)^{1/n}\\ 
    &\qquad\ge \frac{7}{16}  - \frac{m\left(\frac{n}{n+1}\right)^n + \sqrt{m \log n}}{m-2 - 2\sqrt{m \log n}}  \tag{Hoeffding high-probability bound}\\
    &\qquad= \frac{7}{16} - \frac{\left(\frac{n}{n+1}\right)^n + \sqrt{ \frac{\log n}{m}}}{1-\frac{2}{m}- 2\sqrt{\frac{\log n}{m}}}.
\end{align*}

Let $a_n:=\left(\frac{n}{n+1}\right)^n$ and $\delta_n:=\sqrt{\frac{\log n}{m}}$. Since $m\in\Omega(n)$, we have $\delta_n\to0$, while $a_n\to e^{-1}<3/8$. Therefore, for all sufficiently large $n$,
\[
\frac{a_n+\delta_n}{1-\frac{2}{m}-2\delta_n}\le\frac38.
\]
Consequently,
\[
\sum_{j \in \items}
\frac{(\normc_{i,j}-H_j)^2}{\normc_{i,j}}
\ge \frac{7}{16}-\frac38
=\frac1{16}.
\]
\end{proof}

\section{Allocating Goods to Similar Agents}\label{sec: similar}

So far, we have considered the situation where agents' valuation vectors are dissimilar and proved that fair allocations exist. It is well understood that when agents are identical, EFX allocations exist. As a proof of concept, in this section we prove that when agents are similar --- but not identical --- tEFX allocations of goods exist, complementing our previous results. In this section, we treat agents individually and do not require agents with identical valuations to receive identical bundles.

We consider the fairness notion of \emph{tEFX}, or ``transfer EFX''. An integral allocation $A$ is tEFX if, for all $i, i' \in \agents$, for all goods $j \in A_{i'}$, it holds that $v_i(A_i) + v_{i,j} \geq v_i(A_{i'}) - v_{i,j}$. Note that this notion is stronger than (i.e., it implies) EF2x, which demands that the removal of any two distinct goods eliminates envy, and was recently shown to exist for four additive agents~\cite{ashuri2025ef2x}. Here, we prove that tEFX allocations exist for any number of agents that are ``pairwise similar.'' Formally, we require a bound on the TV distance between $\normv_i$ and $\normv_{i'}$ for any pair of agents $i, i' \in \agents$.


We use the following greedy algorithm. \Cref{algo: greedy} processes items in decreasing order of desirability to agent $1$ and assigns each item to the agent whose current bundle has the smallest value from agent $1$'s point of view; note that if $v_{1}(S) \geq v_{1}(S')$, then $\normv_{1}(S) \geq \normv_{1}(S')$, for any two bundles $S$ and $S'$. This algorithm produces an EFX allocation for identical agents. Here we prove that it also produces a tEFX allocation for non-identical, but ``similar,'' agents.

\begin{algorithm}[h]
\caption{Greedy Allocation}\label{algo: greedy}
\begin{algorithmic}
\STATE Initialize $A = (A_1, \dots, A_n)$, with $A_i = \emptyset$.
\STATE Rename the items in nonincreasing order of $v_{1,j}$ (i.e., $v_{1,1} \geq v_{1,2} \geq \dots$).
\FOR{each item $j$ in the sorted order}
    \STATE Choose any $i \gets \arg\min_{i' \in \agents} v_1(A_{i'})$
    \STATE $A_i \gets A_i \cup \{j\}$
\ENDFOR
\end{algorithmic}
\end{algorithm}

\begin{theorem}\label{thm: tefx}
    Let $D^* = \max_{i, i' \in \agents} D_{TV}(\normv_i || \normv_{i'})$. If $\normv_{min} \geq 8 D^*$, then a tEFX allocation exists.
\end{theorem}

\begin{proof}
Let $A$ be the allocation of \Cref{algo: greedy}, and let $i, i' \in \agents$ be any pair of agents such that $\normv_1(A_i) \ge \normv_1(A_{i'})$.  By the definition of TV, $|\normv_1(A_i)-\normv_i(A_i)| \le 2 D_{TV}( \normv_1 || \normv_i) \le 2 D^*$, and $|\normv_1(A_{i'})-\normv_i(A_{i'})| \le 2 D^*$. Combining the inequalities, we have that $\normv_i(A_{i'}) - \normv_i(A_i) = \normv_i(A_{i'}) + \left( \normv_1 (A_{i'}) - \normv_1 (A_{i'})  \right) + \left( - \normv_1 (A_{i}) + \normv_1 (A_{i})\right)  - \normv_i(A_i) \le 4 D^* \leq \frac{\normv_{min}}{2}$; therefore, since envy is scale-free, the removal of any item from $A_{i'}$ eliminates $i$'s envy.

Next, if $A_i=\emptyset$, the tEFX condition for $i'$ toward $i$
is vacuous. Otherwise, let $j$ be the last item $i$ received in~\Cref{algo: greedy}; then, it must be that $\normv_1(A_i) - \normv_{1,j} \le \normv_1(A_{i'})$. By the definition of TV, we have $|\normv_1(A_i)-\normv_{i'}(A_i)| \le 2D^*$ and $|\normv_1(A_{i'})-\normv_{i'}(A_{i'})| \le 2D^*$. Combining, we have that $\normv_{i'}(A_{i}) - \normv_{i'}(A_{i'}) \le 4 D^* + \normv_{1,j}$. Let $j^* \in \arg\min_{\ell \in A_{i}} \normv_{i',\ell}$. By the definition of TV, we have $| \normv_{i',j^*} - \normv_{1,j^*} | \leq 2D^*$ and $| \normv_{i',j} - \normv_{1,j} | \leq 2D^*$. Since $\normv_{i',j^*} \leq \normv_{i',j}$ ($j \in A_i$ and $j^* \in \arg\min_{\ell \in A_{i}} \normv_{i',\ell}$), and $\normv_{1,j} \leq \normv_{1, j^*}$ ($j^*$ was included no later than $j$ in the bundle $A_i$), we have that $\normv_{1,j} \leq 4D^* + \normv_{i',j^*}$.
Combining, we have that $\normv_{i'}(A_{i}) - \normv_{i'}(A_{i'}) \le 8 D^* + \normv_{i',j^*} \leq \normv_{min} + \normv_{i',j^*}$. Therefore, transferring any item from $A_i$ to $A_{i'}$ eliminates the envy of $i'$.
\end{proof}

\bibliographystyle{alpha}
\bibliography{references}

\end{document}